\documentclass[11pt]{amsart}

\usepackage{graphicx}%
\usepackage{multirow}%
\usepackage{amsmath,amssymb,amsfonts}%
\usepackage{amsthm}%
\usepackage{geometry}
\usepackage{mathtools}%
\usepackage{mathrsfs}%
\usepackage[title]{appendix}%
\usepackage{xcolor}%
\usepackage{textcomp}%
\usepackage{manyfoot}%
\usepackage{booktabs}%
\usepackage{enumitem}%
\usepackage{hyperref}

\newtheorem{theorem}{Theorem}[section]%
\newtheorem{proposition}[theorem]{Proposition}%
\newtheorem{lemma}[theorem]{Lemma}%
\newtheorem{corollary}[theorem]{Corollary}%
\newtheorem{example}[theorem]{Example}%
\newtheorem{remark}[theorem]{Remark}%
\newtheorem{definition}[theorem]{Definition}%

\newcommand{\R}{\mathbb{R}}
\newcommand{\C}{\mathbb{C}}
\newcommand{\Hil}{\mathcal{H}}
\newcommand{\Bop}{\mathcal{B}}
\newcommand{\Dop}{\mathcal{D}}
\newcommand{\Tr}{\operatorname{Tr}}
\newcommand{\supp}{\operatorname{supp}}
\newcommand{\1}{\mathbb{1}}
\newcommand{\ket}[1]{\lvert #1 \rangle}
\newcommand{\bra}[1]{\langle #1 \rvert}
\newcommand{\braket}[2]{\langle #1 \vert #2 \rangle}
\newcommand{\proj}[1]{\lvert #1 \rangle\!\langle #1 \rvert}
\newcommand{\QBCE}{\mathrm{QBCE}}
\newcommand{\BCE}{\mathrm{BCE}}
\newcommand{\pp}{\mathrm{pp}}
\newcommand{\Ntot}{n}
\newcommand{\suff}{\succeq}
\newcommand{\garb}{\succeq_{\mathrm{loc}}}
\newcommand{\garbsr}{\succeq_{\mathrm{LOSR}}}
\newcommand{\garbe}{\succeq_{\mathrm{LOSE}}}
\newcommand{\qis}{\succeq_{\mathrm{QIS}}}
\newcommand{\inc}{\sqsupseteq_{\mathrm{inc}}}

\begin{document}
\title[Quantum Bayes Correlated Equilibrium and Quantum Information Structures]{Quantum Bayes Correlated Equilibrium and the Comparison of Quantum Information Structures in Games}
\author{Furkan Sezer*}\thanks{*Texas A\&M University, email: furkan.sezer@tamu.edu.}

\begin{abstract}
Bergemann and Morris (2016) show that one information structure is more informative than another exactly when it induces a smaller set of Bayes correlated equilibrium outcomes in every game. We build the quantum analogue. An information structure becomes a family of density operators indexed by the payoff state, which the mediator observes. We show that obedience is equivalent to a Loewner domination between operators on one player's subsystem. The equilibrium set is then a nonempty compact spectrahedron computable by semidefinite programming, classical structures embed exactly, and under quantum individual sufficiency more information shrinks the equilibrium set in every game.
\end{abstract}


\maketitle

\section{Introduction}\label{sec:intro}

A game of incomplete information decomposes into a \emph{basic game}, comprising the actions, the
payoffs and a prior over payoff states, and an \emph{information structure}
describing what each player learns before taking an action. Bergemann and Morris~\cite{BM2016} (BM)  
identify the order on information structures that governs equilibrium behaviour in such
games. One structure is \emph{individually sufficient} for another when, in some
combined structure, each player's signal under the first is a sufficient statistic for
the state and for the other players' signals under the second. Their theorem is that
this statistical relation holds precisely when the first structure is \emph{more
incentive constrained} than the second, meaning that it induces a subset of Bayes
correlated equilibrium (BCE) outcomes in every basic game. With a single player the order
collapses to Blackwell's~\cite{Blackwell1953}. With several it is strictly coarser,
because what is being compared is no longer the value of information to a decision
maker but the tightness of the obedience constraints that information imposes on a
mediator.

The result matters because it converts a question about behaviour into a question about
statistics. Bayes correlated equilibrium characterizes what can occur when players
observe at least a given information structure and possibly more~\cite[Thm.~1]{BM2016},
so properties shared by all such equilibria are robust predictions, valid whatever
additional information the analyst has failed to model. The ordering then says exactly
when one modelling assumption is more permissive than another, and it is the engine
behind the information-design program built on
it~\cite{BM2016b,BM2019,KG2011}.

This paper develops the quantum analogue: quantum Bayes correlated equilibrium (QBCE). The motivating situation is one in which the
players' private information is carried by quantum systems: each agent holds a
subsystem of a joint state whose preparation depends on an underlying payoff-relevant
parameter, and each may perform any measurement on their own share before acting. Such
descriptions arise naturally for agents distributed across a quantum network, for
sensing platforms in which measurement choices are strategic, and more generally
wherever the information available to a decision maker is a quantum state rather than a
classical signal. The single-player instance of the comparison problem is settled:
Buscemi's comparison of quantum statistical models~\cite{Buscemi2012} extends the
Blackwell--Sherman--Stein theorem to families of density operators, with the physically
implementable transformations of quantum states in place of the classical garblings, and related formulations appear
in Shmaya~\cite{Shmaya2005}, Petz~\cite{Petz1986} and the resource theory of asymmetric
distinguishability~\cite{WangWilde2019}. The multi-player, obedience-constrained order
is the object we construct here.

\subsection{Why the existing quantum literature does not already answer this}

There is a substantial body of work on quantum equilibria, and it divides in a way that
matters here. One line quantizes games of \emph{complete} information: Meyer's quantum
strategies~\cite{Meyer1999} and the protocols of
Eisert--Wilkens--Lewenstein~\cite{EWL1999} and Marinatto--Weber~\cite{MW2000} enlarge
the strategy sets available in a fixed payoff matrix, and correlated equilibria with
quantum signals~\cite{LaMura2005} and quantum strategic games~\cite{Zhang2012} study
which correlations such enlargements sustain. In these settings there is no state of
nature, hence no information structure, and the question of comparing information does
not arise.

The second line treats games of \emph{incomplete} information, where a payoff-relevant
state exists and the players hold private data about it. It includes the
belief-invariant hierarchy of
Auletta--Ferraioli--Rai--Scarpa--Winter~\cite{AFRSW2021}, welfare improvements from
quantum resources~\cite{Groisman2020,AMP2024}, and conflicting-interest Bayesian
games~\cite{Pappa2015,RaiPaul2017}. Here the comparison question is meaningful, and here a
single convention is in force, which the survey of Khan et al.~\cite{Khan2018} states
plainly: the players may share advice drawn before play, but that advice may not depend
on the state of nature. The shared resource, whether a random variable or an entangled
state, is prepared once and independently of the payoff-relevant parameter, so it can
correlate the players' actions with one another but cannot inform them about the
parameter. Under that restriction the achievable correlations take the
local-hidden-variable form from which Bell inequalities are derived, and quantum
advantage in the setting is a Bell violation. Section~\ref{sub:statecorr} states the
restriction formally: the joint conditional behaviour must factor through shared advice
whose law is fixed independently of the state of nature.

Bayes correlated equilibrium removes exactly this restriction. The BM mediator is
\emph{omniscient}, in that a decision rule assigns action recommendations as a function
of the players' signals and of the state of nature itself, and \cite[Sec.~1.2]{BM2016} is explicit that where the
earlier correlated-equilibrium literature imposed \emph{join feasibility}, under which play
depends only on the players' combined information, ``our different motivation leads us
to allow such unexplained correlation'' with the state of nature. Two consequences
follow. A quantum equilibrium notion built from measurements of an
state-independent shared state is a \emph{belief-invariant} object, not a BCE: it is
a fixed-resource correlation concept rather than an information concept. In addition, any
comparison-of-information program must work with ensembles of states that vary with
the state of nature, together with a mediator that observes it. This is the setting we
adopt, and it explains why the ordering question has remained invisible: the established convention forecloses it.

\subsection{The framework and what it yields}

We therefore take a quantum information structure to be a family of density operators
indexed by the payoff state, one subsystem per player, and we let the mediator observe
that state. A decision rule splits each density operator into one positive operator for
every profile of action recommendations, the pieces summing to the original, so that
the players' subsystems are left intact while the recommendation may be correlated both
with the payoff state and with the ensemble decomposition. Obedience requires the
recommended action to be optimal against every alternative measurement a player might
perform on their own subsystem. Realizability places no constraint at all: every joint
distribution over states and actions with the correct marginal is induced by some
decision rule, so the theory is about incentives and nothing else, exactly as in the
classical case.

The technical engine of the paper is that this obedience condition has a closed
operator form. For each player, recommendation and candidate deviation there is a
Hermitian obedience operator on that player's subsystem, and obedience holds precisely
when the operator attached to the recommended action dominates, in the Loewner order,
the operator attached to every alternative. A quantifier over a continuum of positive
operator-valued measures is thereby replaced by finitely many operator inequalities,
and the resulting object is the exact operator analogue of the coordinatewise
comparison that defines classical obedience. Three consequences are immediate: the
equilibrium set is a spectrahedron, hence convex and compact, so equilibrium
computation and information design are semidefinite programming; the set is nonempty,
since ignoring the quantum systems and playing a state-wise Nash equilibrium is always
obedient; and classical information structures embed \emph{exactly}, the quantum
equilibrium set at an embedded classical structure coinciding with the classical Bayes
correlated equilibrium set rather than merely containing it.

The comparison theorem follows the same route. We define quantum individual sufficiency
by requiring a transfer of decision rules that may depend on the payoff state, together
with a single local map per player that does not, mirroring the two operative features
of the classical definition. The relation is a preorder, it contains local garblings,
garblings assisted by shared randomness and garblings assisted by shared entanglement,
and it restricts to Bergemann--Morris individual sufficiency on classical structures.
Under it, more information shrinks the set of equilibrium outcomes in every basic game
and for every number of players. We establish the converse in two regimes: for
structures simultaneously diagonalizable in a product basis, where the classical
theorem is recovered through the exactness of the embedding, and for every comparison
involving the least informative structure, where the criterion is that each player's
local marginal be independent of the payoff state.

Two further results fix the position of the framework. The access model in which
players receive only classical advice extracted from their quantum systems is
degenerate: its equilibrium set does not depend on the quantum structure at all, so the
comparative content of a quantum information structure resides entirely in the local
quantum deviations it permits. Also, the solution concept has the same foundational
justification as its classical counterpart, being identified with the quantum Bayes Nash
equilibria of a canonical expansion under arbitrary local deviations.

\subsection{Organization}

Section~\ref{sec:related} places the results against the literature and lists them.
Section~\ref{sec:classical} fixes the classical apparatus. Sections~\ref{sec:qis}
and~\ref{sec:qbce} define quantum information structures and quantum Bayes correlated
equilibrium. Section~\ref{sec:loewner} proves the Loewner characterization and its
immediate consequences. Section~\ref{sec:struct} establishes the exact recovery of the
classical theory, the degeneracy of the advice model, and the criterion for a structure
to carry no incentive content. Section~\ref{sec:bm1} gives the expansion theorem,
identifying the equilibria with the quantum Bayes Nash equilibria of a canonical expansion.
Section~\ref{sec:order} then defines quantum individual sufficiency and proves the
comparison theorem, together with the two regimes in which its converse holds.
Section~\ref{sec:example} computes a welfare frontier along a chain of increasingly
noisy quantum structures.

\section{Related work and contributions}\label{sec:related}
This work sits between two literatures that have developed largely independently: the
study of quantum games, and the study of orderings on information structures in
economics. We review them in five groups, in the order in which they bear on what
follows. The first supplies the single-agent boundary condition that our results must
reduce to. The second and third are the quantization protocols for complete-information
games and the Bayesian quantum games built around Bell inequalities, which together
account for the modelling convention identified in Section~\ref{sec:intro}. The fourth
contains the closest relatives in solution concept. The fifth is the classical theory
whose ordering this paper quantizes.
\subsection{Review of the related literature}
\paragraph{Quantum single-agent sufficiency.} Buscemi~\cite{Buscemi2012} settles the
quantum Blackwell order, using tests with a reference system;
Shmaya~\cite{Shmaya2005}, Petz~\cite{Petz1986} and the asymmetric-distinguishability
resource theory~\cite{WangWilde2019} give adjacent formulations. These provide the
single-agent boundary condition (Corollary~\ref{cor:consistency}), not a target. The
older stratum of quantum statistical decision theory on which they rest is the theory
of optimal measurement for discriminating a known family of states, due to
Helstrom~\cite{Helstrom1976}, Holevo~\cite{Holevo1973} and Yuen, Kennedy and
Lax~\cite{YKL1975}; the optimal discrimination bound quoted in
Example~\ref{ex:helstrom} is theirs.

\paragraph{Quantization protocols and their critics.} Meyer~\cite{Meyer1999},
Eisert--Wilkens--Lewenstein~\cite{EWL1999} and Marinatto--Weber~\cite{MW2000} supply
the dominant protocols; van Enk and Pike~\cite{vanEnkPike2002} and the discussion
in~\cite{Khan2018} question how much of the advantage is genuinely quantum. None of
this work orders information structures.

\paragraph{Bayesian quantum games and Bell inequalities.} Brunner and
Linden~\cite{BrunnerLinden2013} connect nonlocality to Bayesian games through the
local-hidden-variable form of shared advice; Cheon and Iqbal~\cite{CheonIqbal2008} and Fine~\cite{Fine1982} supply
the factorizability apparatus; Pappa et al.~\cite{Pappa2015} and Rai and
Paul~\cite{RaiPaul2017} treat conflicting-interest games. The standing assumption
throughout is that advice is independent of the state of nature, which
Definition~\ref{def:qdr} removes.

\paragraph{Quantum correlated and belief-invariant equilibria.} Auletta et
al.~\cite{AFRSW2021} organize communication, belief-invariant and correlated equilibria
for games of incomplete information; this is the closest relative in solution concept,
but with a state-independent shared resource and no order on information structures.
La Mura~\cite{LaMura2005} and Brandenburger and La Mura~\cite{BrandenburgerLaMura2016}
study correlated equilibria and team decisions with quantum signals;
Zhang~\cite{Zhang2012} develops quantum strategic game theory. Groisman et
al.~\cite{Groisman2020} and Abbott--Mhalla--Pocreau~\cite{AMP2024} study welfare under
quantum resources, the latter distinguishing direct quantum access from classical
advice. That distinction is mirrored here at the equilibrium level and shown degenerate in
one direction (Theorem~\ref{thm:ppdegen}). 

\paragraph{Classical information design and orderings.} The BCE apparatus is due to
Bergemann and Morris~\cite{BM2016,BM2016b,BM2019}, building on Aumann~\cite{Aumann1987}
and Forges~\cite{Forges1993,Forges2006}. Adjacent orderings are studied by
Gossner~\cite{Gossner2000} for Bayes Nash equilibrium, Lehrer, Rosenberg and
Shmaya~\cite{LRS2010,LRS2013} for garbling and outcome equivalence, and
Liu~\cite{Liu2015} for belief-invariant BCE; the one-player null-structure case is
Kamenica and Gentzkow~\cite{KG2011}.

\subsection{Summary of contributions}

\begin{enumerate}[leftmargin=1.5em]
  \item \textbf{Definitions} (Sections~\ref{sec:qis}--\ref{sec:qbce}). A quantum
  information structure is a state-correlated ensemble; a quantum decision rule is an
  decomposition of each state-conditional density operator into one positive operator
  per profile of recommendations, the pieces summing to the original;
  QBCE requires obedience against arbitrary local positive operator-valued measure (POVM) deviations. Realizability is
  unconstrained (Lemma~\ref{lem:realize}), so only incentives bind, exactly as for BCE.
  \item \textbf{The Loewner characterization} (Theorem~\ref{thm:loewner}): obedience
  holds if and only if the obedience operator attached to the recommended action
  dominates, in the Loewner order, the operator attached to every alternative. This
  replaces a quantifier over a continuum of POVMs by finitely many operator
  inequalities and drives everything that follows.
  \item \textbf{Well-posedness}: the equilibrium set is a nonempty, compact, convex projected
  spectrahedron and membership is a semidefinite program
  (Corollary~\ref{cor:spectra}, Proposition~\ref{prop:nonempty}).
  \item \textbf{Exact classical recovery} (Theorem~\ref{thm:classical}):
  the quantum equilibrium set at an embedded classical structure coincides with the
  classical one. The coherent quantum mediator produces no outcomes
  beyond the classical ones on classical information.
  \item \textbf{Degeneracy of the advice model} (Theorem~\ref{thm:ppdegen}) with an
  explicit separation (Example~\ref{ex:helstrom}): under classical post-processing the
  equilibrium set does not depend on the quantum structure at all, so only the
  quantum-deviation model can support a comparison theorem.
  \item \textbf{The bottom of the order} (Theorem~\ref{thm:trivial}): a structure is
  incentive-equivalent to the null structure exactly when every local marginal is independent of the state of nature, so all incentive content is
  carried by the way those marginals vary with it.
  \item \textbf{Foundations of the solution concept} (Theorem~\ref{thm:bm1},
  Corollary~\ref{cor:bm1general}): QBCE decision rules are exactly the quantum Bayes
  Nash equilibria of a canonical expansion, under arbitrary local deviations. The
  mediator itself admits no representation as a quantum instrument
  (Proposition~\ref{prop:noinstrument}).
  \item \textbf{Why the relation uses positive maps} (Theorem~\ref{thm:transpose}): an
  ensemble and its transposes are incentive-equivalent yet related by no channel, so the
  incentive ordering cannot be characterized by completely positive maps.
  Definition~\ref{def:qis-order} is calibrated accordingly.
  \item \textbf{The comparison theorem, soundness half}
  (Definition~\ref{def:qis-order}, Theorem~\ref{thm:soundness}): quantum individual
  sufficiency implies inclusion of the equilibrium sets, in every basic game and for
  every number of players.
  The relation is a preorder (Proposition~\ref{prop:preorder}), and local,
  shared-randomness and shared-entanglement garblings, as well as classical individual
  sufficiency, are instances (Propositions~\ref{prop:instances},
  \ref{prop:classicalIS}).
  \item \textbf{Completeness} (Theorems~\ref{thm:classicalcomplete}
  and~\ref{thm:bottomcomplete}): the converse of Theorem~\ref{thm:soundness} holds for
  structures simultaneously diagonalizable in product bases, and for every comparison
  involving the least informative structure.
  \item \textbf{A worked comparison} (Section~\ref{sec:example}): along a depolarizing
  chain of quantum structures the QBCE welfare frontier is computed by semidefinite
  programming and falls monotonically with informativeness, the quantum form of the
  Bergemann--Morris comparative static.
\end{enumerate}

\section{Classical preliminaries}\label{sec:classical}

This section fixes the classical apparatus that the rest of the paper quantizes. Two
results of Bergemann and Morris are recorded, and they play different roles. The first
justifies the solution concept, characterizing Bayes correlated equilibrium as the
behaviour consistent with players observing at least a given information structure and
possibly more; it is the reason the concept is the right one to compare structures by.
The second is the comparison theorem itself, equating a statistical ordering on
information with an incentive ordering on behaviour. Sections~\ref{sec:bm1}
and~\ref{sec:order} construct the quantum analogues in that order.

Fix a finite set $N=\{1,\dots,\Ntot\}$ of players, writing $\Ntot$ for their number, and a finite set $\Omega$ of states of nature $\omega$. We define a basic game as follows:
\begin{definition}[Basic game]
$G=(\Omega,(A_i)_{i\in N},(u_i)_{i\in N},\psi)$ with finite action sets $A_i$,
$A:=\prod_iA_i$, payoffs $u_i:A\times\Omega\to\R$, and a full-support prior
$\psi\in\Delta(\Omega)$.
\end{definition}

For a given a basic game $G$, classical information structures are defined as follows:
\begin{definition}[Classical information structure]
$S=((T_i)_{i\in N},\pi)$ with finite signal sets $T_i$, $T:=\prod_iT_i$, and
$\pi:\Omega\to\Delta(T)$. The \emph{null structure} $S_\emptyset$ has every $T_i$ a
singleton.
\end{definition}

A \emph{decision rule} is $\sigma:T\times\Omega\to\Delta(A)$; its \emph{outcome} is
$\nu\in\Delta(\Omega\times A)$ with
$\nu(\omega,a)=\psi(\omega)\sum_t\pi(t\mid\omega)\sigma(a\mid t,\omega)$.

\begin{definition}[Obedience and BCE, \cite{BM2016}]\label{def:bce}
$\sigma$ is \emph{obedient} for $(G,S)$ if for each $i$, $t_i\in T_i$, $a_i\in A_i$ and
$b_i\in A_i$,
\[
\sum_{\omega,t_{-i},a_{-i}}\psi(\omega)\pi(t_i,t_{-i}\mid\omega)\sigma(a_i,a_{-i}\mid t,\omega)
\big[u_i(a_i,a_{-i},\omega)-u_i(b_i,a_{-i},\omega)\big]\ \ge\ 0,
\]
and is a \emph{Bayes correlated equilibrium} if obedient. $\BCE(G,S)$ denotes the set
of BCE outcomes.
\end{definition}

Two features of Definition~\ref{def:bce} carry the weight of the theory and are worth
isolating, because the quantum framework preserves both.

The mediator is \emph{omniscient}: the decision rule $\sigma(a\mid t,\omega)$ conditions
on the state of nature, not only on the players' signals. Much of the earlier
literature on correlated equilibrium with incomplete information instead imposes
\emph{join feasibility}, requiring $\sigma(a\mid t,\omega)$ to be independent of
$\omega$, so that the mediator knows no more than the players collectively
do~\cite{Forges1993,Forges2006}. A second common restriction is \emph{belief
invariance}, requiring that a player's recommendation not alter their beliefs about
the state and the other players' types~\cite{Forges2006,Liu2015}. Bergemann and Morris
impose neither, and say so explicitly: where the earlier literature restricts the
mediator, ``our different motivation leads us to allow such unexplained
correlation''~\cite[Sec.~1.2]{BM2016}.

Obedience is the \emph{only} constraint. Nothing requires the outcome to be
implementable from the signals alone; a decision rule may induce any joint distribution
over states and actions that the obedience inequalities permit. Information therefore
enters as an incentive constraint rather than as a feasibility constraint, and it is
this asymmetry that makes more information shrink the equilibrium set rather than
enlarge it. In Bayes Nash equilibrium the two roles are conflated, which is why the
corresponding ordering, due to Gossner~\cite{Gossner2000}, ranks structures by
correlation possibilities as well as by information.

Writing $V_i(a_i,b_i,t_i)$ for the expected payoff to player $i$ of type $t_i$ who is
recommended $a_i$ and plays $b_i$, obedience reads $V_i(a_i,a_i,\cdot)\ge
V_i(a_i,b_i,\cdot)$: a \emph{coordinatewise} comparison of two functions of $t_i$.
Section~\ref{sec:loewner} reproduces it with ``coordinatewise'' replaced by
``Loewner''.

An \emph{expansion} of $S=((T_i),\pi)$ is an information structure
$S^\ast=((T_i\times T_i'),\pi^\ast)$ for some $((T_i'),\pi')$, with $\pi^\ast$ a joint
law on $\Omega\times T\times T'$ whose $(\Omega,T)$-marginal is $\pi$; players observe
their signal under $S$ together with an additional signal. A strategy profile
$\beta$ for $(G,S^\ast)$ \emph{induces} the decision rule $\sigma$ for $(G,S)$ obtained
by averaging out the additional signals. We are now ready to state Theorem 1 of BM \cite{BM2016}:
\begin{theorem}[{\cite[Thm.~1]{BM2016}}]\label{thm:bm1c}
A decision rule $\sigma$ is a Bayes correlated equilibrium of $(G,S)$ if and only if
there is an expansion $S^\ast$ of $S$ and a Bayes Nash equilibrium of $(G,S^\ast)$
that induces $\sigma$.
\end{theorem}

Theorem~\ref{thm:bm1c} provides the foundational justification for using BCE to address the comparison of information structures. It says that the Bayes correlated equilibria of $(G,S)$ are exactly the
behaviours consistent with the players having observed \emph{at least} the signals in
$S$, whatever further information the analyst has failed to model. Properties shared by
all of them are therefore robust to that ignorance, and an ordering built on
$\BCE(G,\cdot)$ compares information structures without committing to a particular
account of what else the players know. The omniscience of the mediator is the device
that achieves this: an arbitrary correlating mechanism stands in for arbitrary
unmodelled information. Section~\ref{sec:bm1} proves the quantum analogue.

We turn to the ordering. In the one-player case an information structure is an
experiment, and Blackwell's theorem compares experiments by garbling. The many-player
generalization replaces garbling by a player-by-player conditional independence
requirement inside a combined structure.

\begin{definition}[Individual sufficiency, {\cite[Def.~6]{BM2016}}]\label{def:is}
$S$ is \emph{individually sufficient} for $S'$, written $S\suff S'$, if there is a
combined information structure $\pi^\ast:\Omega\to\Delta(T\times T')$, that is, a law
whose $T$-marginal is $\pi$ and whose $T'$-marginal is $\pi'$, such that for each player
$i$ the conditional probability
\[
\Pr\big(t_i'\mid t_i,t_{-i},\omega\big)
\;=\;\frac{\sum_{t_{-i}'}\pi^\ast\big((t_i,t_{-i}),(t_i',t_{-i}')\mid\omega\big)}
{\pi(t_i,t_{-i}\mid\omega)}
\]
is independent of $t_{-i}$ and $\omega$, whenever the denominator is nonzero. The
denominator is the stated one because summing $\pi^\ast$ over all of $T'$ returns its
$T$-marginal.
\end{definition}

Writing $\varphi_i(t_i'\mid t_i)$ for the common value of that conditional, individual
sufficiency is equivalent to the existence of Markov kernels
$\varphi_i:T_i\to\Delta(T_i')$ satisfying
\begin{equation}\label{eq:phi}
\varphi_i(t_i'\mid t_i)\,\pi(t_i,t_{-i}\mid\omega)
=\sum_{t_{-i}'}\pi^\ast\big((t_i,t_{-i}),(t_i',t_{-i}')\mid\omega\big)
\qquad\text{for all }t_i',t_i,t_{-i},\omega,
\end{equation}
which is the form used by Bergemann and Morris~\cite[eq.~(9)]{BM2016} and the form we
quantize. Note that \eqref{eq:phi} needs no side condition: where
$\pi(t_i,t_{-i}\mid\omega)=0$ both sides vanish, since the right-hand side is bounded
above by the $T$-marginal of $\pi^\ast$.

\begin{theorem}[{\cite[Thm.~2]{BM2016}}]\label{thm:bm}
$S\suff S'$ if and only if $\BCE(G,S)\subseteq\BCE(G,S')$ for every basic game $G$ on
the common basic structure. For $\Ntot=1$ the order $\suff$ is Blackwell's.
\end{theorem}

Theorem~\ref{thm:bm} is the result whose quantum analogue this paper constructs, so it
is worth being explicit about which of its features the construction relies on and
which it does not.

The soundness half rests on two properties of the transfer $\varphi_i$, and both survive
quantization. The incentive constraints under $S'$ are $\varphi_i$-averages of those
under $S$~\cite[eq.~(10)]{BM2016}, so obedience is transported by averaging rather than
by any construction peculiar to classical signals; and the kernels in \eqref{eq:phi}
depend on neither $\omega$ nor $t_{-i}$, which is the whole content of individual
sufficiency. Section~\ref{sec:order} reproduces both. Averaging becomes the application
of a positive linear map, which is the operation that preserves the order in which
quantum obedience is expressed; and the conditional independence requirement becomes
the demand that one local map serve for every state of nature.

The equivalence in Theorem~\ref{thm:bm} is exact, the statistical and incentive
orderings coinciding on every pair of structures. In quantum systems, we established that the soundness half holds in general (Theorem \ref{thm:soundness}), and the converse is established on two classes of pairs (Theorems \ref{thm:classicalcomplete}-\ref{thm:bottomcomplete}).

One feature of the classical order is easy to misread and is worth stating plainly,
because the quantum framework inherits it. Individual sufficiency ranks structures by
beliefs and higher-order beliefs about the state alone. A structure and its expansion
by a correlating device that conveys nothing about $\Omega$ are therefore equivalent in
the order, even though the expansion allows the players to coordinate in ways the
original does not~\cite[Example~1]{BM2016}. Correlation possibilities are invisible to
the ordering because Bayes correlated equilibrium already grants the mediator arbitrary
correlating power, so adding more changes no obedience constraint. The quantum
counterpart of this observation is recorded in Section~\ref{sec:struct}, where a
structure is shown to be equivalent to the uninformative one exactly when each player's
local marginal does not vary with the state.

\section{Quantum information structures}\label{sec:qis}

For a finite-dimensional Hilbert space $\Hil$ let $\Bop(\Hil)$ denote the bounded
operators and $\Dop(\Hil)$ the density operators. A \emph{POVM} with outcomes in a
finite set $Y$ is a family $\{M^y\}_{y\in Y}$ with $M^y\succeq0$ and $\sum_yM^y=\1$. A
linear map $\Phi:\Bop(\Hil)\to\Bop(\Hil')$ is \emph{positive} if $\Phi(\xi)\succeq0$
whenever $\xi\succeq0$, and \emph{trace preserving} if $\Tr\Phi(\xi)=\Tr\xi$ for all
$\xi$. It is \emph{completely positive} if $\Phi\otimes\mathrm{id}_{\Bop(\C^d)}$ is
positive for every $d\ge1$, that is, if $\Phi$ remains positive when applied to one
half of a system whose other half is left alone. A \emph{channel} is a completely positive
trace-preserving map; these are the transformations a physical process can implement.
Every channel is thus a positive trace-preserving map, but not conversely: transposition
in a fixed basis is positive and trace preserving without being completely positive, so
it is not a channel. The comparison relation of Section~\ref{sec:order} is defined using
positive trace-preserving maps, and not the smaller class of channels.

\begin{definition}[Quantum information structure]\label{def:qis}
A quantum information structure $Q$ is defined as $Q=\big((\Hil_i)_{i\in N},(\rho_\omega)_{\omega\in\Omega}\big)$ with
$\Hil=\bigotimes_i\Hil_i$ and $\rho_\omega\in\Dop(\Hil)$, player $i$ holding
$\Hil_i$. Equivalently $Q$ is the classical--quantum state
$\rho^Q=\sum_\omega\psi(\omega)\proj{\omega}_\Omega\otimes\rho_\omega$, the
$\Omega$-register being held by the analyst and inaccessible to the players. Write
$\rho_{i,\omega}:=\Tr_{-i}\rho_\omega$ for the local marginals.
\end{definition}
To demonstrate how this generalizes the classical framework, we show how a standard classical information structure translates into this notation.
\begin{example}[Embedding of classical structures]\label{ex:embed}
A classical structure $S=((T_i),\pi)$ embeds into the quantum framework via the mapping $\iota(S)$ with $\Hil_i=\C^{T_i}$ and
$\rho_\omega=\sum_{t\in T}\pi(t\mid\omega)\bigotimes_i\proj{t_i}$.
\end{example}

\subsection{State correlation}\label{sub:statecorr}

The dependence of $\rho_\omega$ on $\omega$ is what distinguishes
Definition~\ref{def:qis} from the objects
of~\cite{AFRSW2021,AMP2024,BrunnerLinden2013}, where a \emph{fixed} state is
shared and entanglement correlates actions without carrying information about nature.

The same restriction governs classical advice in that literature, and it can now be
stated exactly. Let $\mathcal{L}$ be a finite set of advice values, let
$p\in\Delta(\mathcal{L})$ be a distribution on it chosen independently of $\omega$, and
consider two players with signals $t_1\in T_1$, $t_2\in T_2$ and actions $a_1\in A_1$,
$a_2\in A_2$. Advice of this kind constrains the joint conditional behaviour to the form
\begin{equation}\label{eq:lhv}
P(a_1,a_2\mid t_1,t_2)=\sum_{\lambda\in\mathcal{L}}p(\lambda)\,
P(a_1\mid t_1,\lambda)\,P(a_2\mid t_2,\lambda),
\end{equation}
following Brunner and Linden~\cite{BrunnerLinden2013}. Equation~\eqref{eq:lhv} is the
local-hidden-variable ansatz from which Bell inequalities are derived: the advice
$\lambda$ is drawn once, before play, and its law $p$ carries no dependence on the
state of nature. Belief-invariant structures are the degenerate case
$\rho_\omega\equiv\rho$ of Definition~\ref{def:qis}. Only with $\omega$-dependence, 
``more informative'' is a meaningful relation between quantum structures.

\subsection{No-signalling is not at issue}\label{sub:nosignalling}

Because \eqref{eq:lhv} is the hypothesis of Bell's theorem, it is worth saying exactly
what is and is not being relaxed. In the nonlocal-game reading of \eqref{eq:lhv} the
players' signals $t_1,t_2$ play the role of measurement settings chosen at spacelike
separation, and the ansatz encodes two physical constraints on that arrangement:
measurement independence, since the law $p$ of the advice does not depend on
$t_1,t_2$; and locality, since $P(a_i\mid t_i,\lambda)$ does not depend on the other
player's setting. Under those constraints a violation of the resulting inequalities is
a physical statement about the world.

Neither constraint is a law governing the objects compared here. In a game of
incomplete information the payoff state $\omega$ is realized before play and is
causally upstream of the signals, which are drawn from $\pi(\cdot\mid\omega)$; it is
not a freely chosen setting, so conditioning a preparation on it is ordinary state
preparation conditioned on a classical parameter, the same object that a quantum
statistical model $\{\rho_\omega\}$ already is~\cite{Buscemi2012}. Nor are the players
assumed to be spacelike separated. Bayes correlated equilibrium answers the question of
what behaviour is consistent with the players observing \emph{at least} the given
information structure and possibly more, and Section~\ref{sec:bm1} makes this precise:
every equilibrium of the concept defined below is realized as a Bayes Nash equilibrium
of an expansion whose state is prepared conditionally on the payoff state and whose
subsystems are then distributed. That is the implementable content of the
mediator, and it involves no transmission of any kind during play.

If instead one wishes to model agents who are spacelike separated and hold exactly the
information in $Q$, the appropriate restriction is belief invariance, which by
Section~\ref{sub:statecorr} is the degenerate case $\rho_\omega\equiv\rho$ of
Definition~\ref{def:qis} and is the setting of~\cite{AFRSW2021,Liu2015}. The present
framework contains that case rather than contradicting it, and the two answer different
questions: the belief-invariant concept fixes the information and varies the
correlating device, while the concept used here fixes the correlating device as
unrestricted and varies the information, which is what makes information comparable.

\subsection{Classical actions and the scope of the linearity hypothesis}\label{sub:classicalactions}

Players measure their subsystems and play classical actions. It is worth being precise
about what this assumes, since linearity is used throughout.

Payoffs enter the analysis only through the outcome $\nu\in\Delta(\Omega\times A)$, as
$\nu\mapsto\sum_{\omega,a}u_i(a,\omega)\nu(\omega,a)$, and $\nu$ is linear in the
decision rule because the trace is linear. The linearity we use is therefore exactly
the linearity of expected utility in the induced distribution over outcomes: no further
restriction is imposed on $u_i$, which may be an arbitrary real function on
$A\times\Omega$. Consequently the framework covers every finite basic game with von
Neumann--Morgenstern payoffs and every finite-dimensional quantum information
structure, which is precisely the scope of~\cite{BM2016} on the classical side, as
Section~\ref{sec:struct} confirms, by showing that the quantum equilibrium set at an
embedded classical structure is exactly the classical one. Convexity and semidefinite representability are
consequences of this, not additional hypotheses.

What lies outside is the class of quantization protocols in which the payoff is a
nonlinear function of the pre-measurement state, as in the strategy spaces of
Eisert--Wilkens--Lewenstein~\cite{EWL1999} and its descendants. That exclusion is not a
loss of generality so much as a separation from a class where the solution concept
itself is fragile: the survey~\cite{Khan2018} records that when payoffs fail to be
linear in the strategy variables even existence of Nash equilibrium can fail. Since the
present object is an ordering defined through equilibrium sets, a setting in which
those sets may be empty is not a viable domain for it.

\section{Quantum Bayes correlated equilibrium}\label{sec:qbce}

\subsection{Decision rules}

The BM mediator observes $\omega$ and issues a private recommendation to each player.
Its quantum counterpart does not touch the players' subsystems; it splits the ensemble.

\begin{definition}[Quantum decision rule]\label{def:qdr}
A \emph{quantum decision rule} for $(G,Q)$ is a family
$\{\rho_\omega^a\}_{\omega\in\Omega,a\in A}$ of positive semidefinite operators on
$\Hil$ with
\[
\sum_{a\in A}\rho_\omega^a=\rho_\omega\qquad\text{for every }\omega\in\Omega .
\]
Player $i$ privately receives the classical recommendation $a_i$ and retains $\Hil_i$.
The induced \emph{outcome} is $\nu(\omega,a)=\psi(\omega)\Tr[\rho_\omega^a]$.
\end{definition}

Conditional on $(a,\omega)$ the players' joint state is
$\rho_\omega^a/\Tr[\rho_\omega^a]$. The recommendation may be correlated both with
$\omega$, as BCE requires, and with the ensemble decomposition of $\rho_\omega$, while
the marginal $\rho_\omega$ is left intact. Note $\rho_\omega^a\preceq\rho_\omega$, so
$\supp\rho_\omega^a\subseteq\supp\rho_\omega$.

\begin{lemma}[Realizability is unconstrained]\label{lem:realize}
For every $\nu\in\Delta(\Omega\times A)$ whose $\Omega$-marginal is $\psi$ there is a
quantum decision rule with outcome $\nu$, namely the \emph{proportional rule}
$\rho_\omega^a:=\big(\nu(\omega,a)/\psi(\omega)\big)\rho_\omega$.
\end{lemma}

\begin{proof}
Each $\rho_\omega^a$ is a nonnegative multiple of a density operator, hence positive
semidefinite; $\psi$ has full support, so the ratio is well defined; and
$\sum_a\rho_\omega^a=\big(\sum_a\nu(\omega,a)/\psi(\omega)\big)\rho_\omega=\rho_\omega$
because $\sum_a\nu(\omega,a)=\psi(\omega)$. The outcome is
$\psi(\omega)\Tr[\rho_\omega^a]=\nu(\omega,a)$.
\end{proof}

Lemma~\ref{lem:realize} is the counterpart of the fact that in BCE only obedience
binds. It makes the comparison problem a question about \emph{incentives} rather than
\emph{attainability}, and it is available only because the mediator sees $\omega$.

\subsection{Obedience}

\begin{definition}[Obedience operator]\label{def:Xop}
For a decision rule $\{\rho_\omega^a\}$, a player $i$, a recommendation $a_i$ and a
candidate action $b_i$, set
\begin{equation}\label{eq:Xop}
X_i^{a_i,b_i}\;:=\;\sum_{\omega\in\Omega}\ \sum_{a_{-i}\in A_{-i}}
\psi(\omega)\,u_i(b_i,a_{-i},\omega)\ \Tr_{-i}\!\big[\rho_\omega^{(a_i,a_{-i})}\big]
\ \in\ \Bop(\Hil_i).
\end{equation}
\end{definition}

Each $X_i^{a_i,b_i}$ is Hermitian, because $\rho_\omega^a$ is positive semidefinite,
the partial trace preserves Hermiticity, and the coefficients are real. If player $i$
obeys, the contribution to $i$'s expected payoff conditional on the recommendation
$a_i$ is $\Tr[X_i^{a_i,a_i}]$; if instead $i$ measures $\Hil_i$ with a POVM
$\{M^{b_i}\}_{b_i\in A_i}$ and plays the outcome while all others obey, the
contribution is $\sum_{b_i}\Tr[M^{b_i}X_i^{a_i,b_i}]$.

\begin{definition}[QBCE]\label{def:qbce}
A quantum decision rule $X_i^{a_i,b_i}$ is a \emph{quantum Bayes correlated equilibrium} of $(G,Q)$ if
for every $i\in N$, every $a_i\in A_i$ and every POVM $\{M^{b_i}\}_{b_i\in A_i}$ on
$\Hil_i$,
\begin{equation}\label{eq:qobed}
\sum_{b_i\in A_i}\Tr\!\big[M^{b_i}X_i^{a_i,b_i}\big]\ \le\ \Tr\!\big[X_i^{a_i,a_i}\big].
\end{equation}
$\QBCE(G,Q)\subseteq\Delta(\Omega\times A)$ denotes the set of QBCE outcomes.
\end{definition}

Choosing $M^{a_i}=\1$ and $M^b=0$ otherwise shows the left side of \eqref{eq:qobed}
always attains $\Tr[X_i^{a_i,a_i}]$; the condition says this choice is optimal.

\begin{definition}[Advice model]\label{def:qbcepp}
$\QBCE_\pp(G,Q)$ is defined by restricting \eqref{eq:qobed} to deterministic
relabellings of the recommendation, i.e.\ to the constraints
$\Tr[X_i^{a_i,b_i}]\le\Tr[X_i^{a_i,a_i}]$ for all $b_i\in A_i$.
\end{definition}

Definitions~\ref{def:qbce} and~\ref{def:qbcepp} are the equilibrium-level images of the
two resource models of Abbott--Mhalla--Pocreau~\cite{AMP2024}: direct quantum access
versus classical advice extracted from a quantum device. Since the relabelling
constraints are a subset of the POVM constraints,
$\QBCE(G,Q)\subseteq\QBCE_\pp(G,Q)$ always. The inclusion is strict, and more than
that: Section~\ref{sec:struct} shows that \\$\QBCE_\pp(G,Q)$ does not depend on $Q$ at
all, so the advice model assigns every quantum information structure the same
equilibrium set and cannot order structures. This is why
Definition~\ref{def:qbce}, and not Definition~\ref{def:qbcepp}, is the concept used
throughout.

\section{The Loewner characterization}\label{sec:loewner}

\begin{definition}[Loewner order]\label{def:loewner}
For Hermitian operators $Z,W$ on a finite-dimensional Hilbert space, write
$Z\succeq W$, and say that $Z$ \emph{Loewner dominates} $W$, if $Z-W$ is positive
semidefinite, that is, if $\bra{v}(Z-W)\ket{v}\ge0$ for every vector $\ket{v}$. The
relation $\succeq$ is a partial order on the real vector space of Hermitian operators,
reflexive, transitive and antisymmetric, and it is the natural order induced by the
cone of positive semidefinite operators. On operators that are diagonal in a common
basis it reduces to comparison of the diagonal entries one by one, so for
$1\times1$ matrices it is the order on the reals.
\end{definition}

Two properties are used repeatedly below and are recorded here. First, if $Z\succeq W$
then $\Tr[FZ]\ge\Tr[FW]$ for every positive semidefinite $F$, since
$\Tr[F(Z-W)]\ge0$ whenever both factors are positive semidefinite. Second, the order is
preserved by every positive linear map $\Phi$, since $Z-W\succeq0$ gives
$\Phi(Z)-\Phi(W)=\Phi(Z-W)\succeq0$; this is what transports obedience in
Section~\ref{sec:order}. The order is partial rather than total: two Hermitian
operators need not be comparable, and Section~\ref{sec:struct} exploits exactly this
in showing that a decision rule can fail obedience because a difference of obedience
operators has eigenvalues of both signs.

\begin{theorem}[Obedience is Loewner domination]\label{thm:loewner}
A quantum decision rule $X_i^{a_i,b_i}$ satisfies \eqref{eq:qobed} for the pair $(i,a_i)$ if and only if
\begin{equation}\label{eq:loewner}
X_i^{a_i,a_i}\ \succeq\ X_i^{a_i,b_i}\qquad\text{for every }b_i\in A_i,
\end{equation}
where $\succeq$ is the Loewner order. Consequently a rule is a QBCE if and only if
\eqref{eq:loewner} holds for every $i$ and every $a_i$.
\end{theorem}

\begin{proof}
($\Leftarrow$) Assume \eqref{eq:loewner} and let $\{M^{b_i}\}$ be a POVM on $\Hil_i$.
For positive semidefinite $F$ and $H$ one has $\Tr[FH]\ge0$; taking $F=M^{b_i}$ and
$C=X_i^{a_i,a_i}-X_i^{a_i,b_i}$ gives
$\Tr[M^{b_i}X_i^{a_i,b_i}]\le\Tr[M^{b_i}X_i^{a_i,a_i}]$ for each $b_i$. Summing and
using $\sum_{b_i}M^{b_i}=\1$,
\[
\sum_{b_i}\Tr\big[M^{b_i}X_i^{a_i,b_i}\big]\le
\Tr\Big[\Big(\sum_{b_i}M^{b_i}\Big)X_i^{a_i,a_i}\Big]=\Tr\big[X_i^{a_i,a_i}\big].
\]

($\Rightarrow$) Assume \eqref{eq:qobed} and suppose \eqref{eq:loewner} fails for some
$b^\ast$, necessarily $b^\ast\ne a_i$. Then $W:=X_i^{a_i,b^\ast}-X_i^{a_i,a_i}$ is
Hermitian with a positive eigenvalue; let $\ket v$ be a corresponding unit eigenvector
and $\delta:=\bra{v}W\ket{v}>0$. Set $M^{b^\ast}:=\proj v$,
$M^{a_i}:=\1-\proj v$ and $M^b:=0$ for $b\notin\{a_i,b^\ast\}$, a POVM. Then
\[
\sum_{b_i}\Tr\big[M^{b_i}X_i^{a_i,b_i}\big]
=\bra{v}X_i^{a_i,b^\ast}\ket{v}+\Tr\big[X_i^{a_i,a_i}\big]-\bra{v}X_i^{a_i,a_i}\ket{v}
=\Tr\big[X_i^{a_i,a_i}\big]+\delta,
\]
contradicting \eqref{eq:qobed}.
\end{proof}

Theorem~\ref{thm:loewner} is the technical hinge of the paper. It replaces a quantifier
over a continuum of POVMs by finitely many operator inequalities, and it identifies
quantum obedience as the exact operator analogue of the classical coordinatewise
condition of Section~\ref{sec:classical}. Three consequences are immediate.

\begin{corollary}[QBCE is a spectrahedron]\label{cor:spectra}
Fix $G$ and $Q$. The set of QBCE decision rules is
\[
\Big\{\{\rho_\omega^a\}:\ \rho_\omega^a\succeq0\ \forall\omega,a;\quad
\textstyle\sum_a\rho_\omega^a=\rho_\omega\ \forall\omega;\quad
X_i^{a_i,a_i}-X_i^{a_i,b_i}\succeq0\ \forall i,a_i,b_i\Big\},
\]
cut out by finitely many linear matrix inequalities and linear equalities in the
variables $\{\rho_\omega^a\}$. It is convex and compact; optimizing a linear objective
over it is a semidefinite program; and $\QBCE(G,Q)$ is convex and compact, being its
image under the linear outcome map.
\end{corollary}

\begin{proof}
By Theorem~\ref{thm:loewner} these are exactly the equilibrium conditions. Each
$X_i^{a_i,b_i}$ is a linear function of $\{\rho_\omega^a\}$ by \eqref{eq:Xop}, since
the partial trace is linear and the coefficients $\psi(\omega)u_i(b_i,a_{-i},\omega)$
are constants; hence each constraint is a linear matrix inequality. A set defined by
finitely many linear matrix inequalities and linear equalities is a spectrahedron,
hence closed and convex, and it is bounded because
$0\preceq\rho_\omega^a\preceq\rho_\omega$ forces $\|\rho_\omega^a\|\le1$. Linear images
of convex compact sets are convex and compact.
\end{proof}

\begin{proposition}[QBCE is nonempty]\label{prop:nonempty}
For every basic game $G$ and every quantum information structure $Q$ we have
$\QBCE(G,Q)\ne\emptyset$. Explicitly, choose for each $\omega$ a mixed Nash equilibrium
$\varsigma_\omega=\bigotimes_i\varsigma_{\omega,i}\in\Delta(A)$ of the
complete-information game $(A,u(\cdot,\omega))$, which exists by Nash's theorem since
$A$ is finite. Then $\rho_\omega^a:=\varsigma_\omega(a)\rho_\omega$ is a QBCE, with
outcome $\nu(\omega,a)=\psi(\omega)\varsigma_\omega(a)$.
\end{proposition}

\begin{proof}
This is a decision rule by Lemma~\ref{lem:realize}. By \eqref{eq:Xop} and
$\Tr_{-i}\rho_\omega=\rho_{i,\omega}$,
\[
X_i^{a_i,a_i}-X_i^{a_i,b_i}
=\sum_\omega\psi(\omega)\,\varsigma_{\omega,i}(a_i)\,\kappa_i(\omega;a_i,b_i)\,\rho_{i,\omega},
\]
where
\[
\kappa_i(\omega;a_i,b_i):=\sum_{a_{-i}}\big(u_i(a_i,a_{-i},\omega)-u_i(b_i,a_{-i},\omega)\big)\varsigma_{\omega,-i}(a_{-i}),
\]
using $\varsigma_\omega(a_i,a_{-i})=\varsigma_{\omega,i}(a_i)\varsigma_{\omega,-i}(a_{-i})$.
If $\varsigma_{\omega,i}(a_i)>0$ then $a_i$ is a best response to
$\varsigma_{\omega,-i}$ in the complete-information game, so
$\kappa_i(\omega;a_i,b_i)\ge0$; if $\varsigma_{\omega,i}(a_i)=0$ the term vanishes.
Hence the displayed operator is a nonnegative combination of the positive semidefinite
operators $\rho_{i,\omega}$, so it is positive semidefinite, and
Theorem~\ref{thm:loewner} applies.
\end{proof}

The rule constructed in Proposition~\ref{prop:nonempty} operates independently of the players' quantum subsystems and
relying solely on a state-wise Nash equilibrium; it is the counterpart of the observation that a
Bayes Nash equilibrium of $(G,S)$ is a BCE of $(G,S)$.

\begin{remark}[Computation]
Corollary~\ref{cor:spectra} places the apparatus inside semidefinite programming:
computing the QBCE-optimal value of a linear welfare objective, or deciding membership
of an outcome, is a single semidefinite program (SDP) whose size is polynomial in $|\Omega|$, $|A|$ and the
dimensions $\dim\Hil_i$. Section~\ref{sec:example} carries this out in a two-player example.
\end{remark}

\section{Recovery of the classical theory, and the advice model}\label{sec:struct}

\begin{theorem}[Exact classical recovery]\label{thm:classical}
For every classical information structure $S$ and every basic game $G$,
\[
\QBCE\big(G,\iota(S)\big)\;=\;\BCE(G,S).
\]
Moreover a diagonal quantum decision rule satisfies \eqref{eq:loewner} if and only if
the associated classical decision rule is obedient in the sense of
Definition~\ref{def:bce}.
\end{theorem}

\begin{proof}
\emph{Step 1: diagonal rules correspond exactly.} Let
$\sigma:T\times\Omega\to\Delta(A)$ and set
$\rho_\omega^a:=\sum_{t}\pi(t\mid\omega)\sigma(a\mid t,\omega)\bigotimes_i\proj{t_i}$,
a decision rule for $\iota(S)$ whose outcome is that of $\sigma$. Each $\rho_\omega^a$
is diagonal in the product basis, hence by \eqref{eq:Xop} so is each $X_i^{a_i,b_i}$;
writing $X_i^{a_i,b_i}=\sum_{t_i}x_{b_i}(t_i)\proj{t_i}$ we obtain
\[
x_{b_i}(t_i)=\sum_{\omega,t_{-i},a_{-i}}
\psi(\omega)\pi(t_i,t_{-i}\mid\omega)\sigma(a_i,a_{-i}\mid t,\omega)u_i(b_i,a_{-i},\omega)
=V_i(a_i,b_i,t_i).
\]
For diagonal Hermitian operators the Loewner order is coordinatewise domination, so
\eqref{eq:loewner} reads $x_{a_i}(t_i)\ge x_{b_i}(t_i)$ for all $t_i$ and $b_i$, which
is Definition~\ref{def:bce}. With Theorem~\ref{thm:loewner} this proves the second
assertion and gives $\BCE(G,S)\subseteq\QBCE(G,\iota(S))$.

\emph{Step 2: general rules reduce to diagonal ones.} Let $\{\rho_\omega^a\}$ be any
QBCE of $(G,\iota(S))$ with outcome $\nu$. Let
$\mathcal{P}_i(\xi):=\sum_{t_i}\proj{t_i}\xi\proj{t_i}$ be the pinching channel and
$\mathcal{P}:=\bigotimes_i\mathcal{P}_i$, and set $\tilde\rho_\omega^a:=\mathcal{P}(\rho_\omega^a)$. Then
$\tilde\rho_\omega^a\succeq0$, and
$\sum_a\tilde\rho_\omega^a=\mathcal{P}(\rho_\omega)=\rho_\omega$ because $\rho_\omega$ is
diagonal and $\mathcal{P}$ fixes diagonal operators; so this is a decision rule, with outcome
$\nu$ since $\mathcal{P}$ is trace preserving, and it is diagonal.

Its obedience operators are $\mathcal{P}_i(X_i^{a_i,b_i})$. Indeed
$\Tr_{-i}\circ\mathcal{P}=\mathcal{P}_i\circ\Tr_{-i}$: by linearity it suffices to check this on product
operators $\xi_i\otimes\xi_{-i}$, where both sides equal $\mathcal{P}_i(\xi_i)\Tr[\xi_{-i}]$
because $\mathcal{P}_{-i}:=\bigotimes_{j\ne i}D_j$ is trace preserving; substituting into
\eqref{eq:Xop} and using linearity of $\mathcal{P}_i$ gives the claim. Since $\mathcal{P}_i$ is positive,
$X_i^{a_i,a_i}\succeq X_i^{a_i,b_i}$ implies
$\mathcal{P}_i(X_i^{a_i,a_i})\succeq \mathcal{P}_i(X_i^{a_i,b_i})$, so $\{\tilde\rho_\omega^a\}$ is a QBCE.
Being diagonal it arises from a classical rule $\sigma$ via
$\sigma(a\mid t,\omega):=\bra{t}\tilde\rho_\omega^a\ket{t}/\pi(t\mid\omega)$ where
$\pi(t\mid\omega)>0$ and arbitrarily elsewhere; by Step 1 that $\sigma$ is a BCE of
$(G,S)$ with outcome $\nu$. Hence $\QBCE(G,\iota(S))\subseteq\BCE(G,S)$.
\end{proof}

A coherent mediator, meaning one that splits a diagonal $\rho_\omega$ into non-diagonal
components, which is possible, for instance
$\tfrac12\1=\tfrac14\begin{psmallmatrix}1&1\\1&1\end{psmallmatrix}+\tfrac14\begin{psmallmatrix}1&-1\\-1&1\end{psmallmatrix}$, therefore
produces no outcomes beyond the classical ones. Coherence in the recommendation buys
extra realizations and costs exactly as much in extra obedience constraints. The
embedding $\iota$ is thus \emph{exact}, and the quantum theory is a conservative
extension of BM. The completeness results of Section~\ref{sec:complete} depend on this
exactness.

\begin{theorem}[The advice model is degenerate]\label{thm:ppdegen}
For every basic game $G$ and every quantum information structure $Q$,
\[
\QBCE_\pp(G,Q)=\Big\{\nu:\ \Xi_i(\nu;a_i,b_i)\ge0\ \ \forall i,a_i,b_i\Big\}=\BCE(G,S_\emptyset),
\]
where $\nu$ ranges over distributions on $\Omega\times A$ with $\Omega$-marginal $\psi$
and
\[
\Xi_i(\nu;a_i,b_i):=\sum_{\omega,a_{-i}}\nu(\omega,a_i,a_{-i})\big[u_i(a_i,a_{-i},\omega)-u_i(b_i,a_{-i},\omega)\big].
\]
In particular $\QBCE_\pp(G,Q)$ does not depend on $Q$, so the order it induces on
quantum information structures is the trivial one in which all structures are
equivalent.
\end{theorem}

\begin{proof}
Composing $\Tr_{-i}$ with $\Tr$ gives the full trace, so by \eqref{eq:Xop}
\[
\Tr\big[X_i^{a_i,b_i}\big]
=\sum_{\omega,a_{-i}}\psi(\omega)u_i(b_i,a_{-i},\omega)\Tr\big[\rho_\omega^{(a_i,a_{-i})}\big]
=\sum_{\omega,a_{-i}}u_i(b_i,a_{-i},\omega)\,\nu(\omega,a_i,a_{-i}),
\]
which depends on the rule only through its outcome. Hence the constraints of
Definition~\ref{def:qbcepp} are exactly $\Xi_i(\nu;a_i,b_i)\ge0$, so every
$\QBCE_\pp$ outcome lies in the displayed set. Conversely, if $\nu$ lies in that set,
Lemma~\ref{lem:realize} realizes it on $Q$ and the identity just proved shows the
constraints hold. For the second equality take $S=S_\emptyset$ in
Definition~\ref{def:bce}: each $T_i$ is a singleton,
$\nu(\omega,a)=\psi(\omega)\sigma(a\mid\omega)$, and the BCE inequalities reduce to
$\Xi_i(\nu;a_i,b_i)\ge0$.
\end{proof}

Under classical post-processing the players never use their quantum systems, so those
systems impose no incentive constraints and carry no comparative content. Any quantum
analogue of Bergemann--Morris must therefore use Definition~\ref{def:qbce}. We record
the separation of the two models concretely.

\begin{example}[Helstrom separation]\label{ex:helstrom}
Let $\Ntot=1$, $\Omega=A=\{0,1\}$, $\psi$ uniform, $u(a,\omega)=\1[a=\omega]$,
$\Hil=\C^2$, and
\[
\rho_0=\proj0,\qquad\rho_1=\proj+,\qquad \ket+=\tfrac1{\sqrt2}(\ket0+\ket1),
\]
a state-correlated ensemble. For fixed $q\in\Delta(A)$ the rule
$\rho_\omega^a:=q(a)\rho_\omega$, whose outcome $\nu(\omega,a)=\tfrac12q(a)$ carries no
information, lies in $\QBCE_\pp$ but not in $\QBCE$ whenever $q(a)>0$.
\end{example}

\begin{proof}
Here $X^{a,b}=\sum_\omega\psi(\omega)u(b,\omega)q(a)\rho_\omega=\tfrac{q(a)}2\rho_b$, so
$\Tr[X^{a,b}]=\tfrac{q(a)}2$ for both $b$ and all relabelling constraints hold with
equality; thus $\nu\in\QBCE_\pp$, as Theorem~\ref{thm:ppdegen} requires. By
Theorem~\ref{thm:loewner}, membership in $\QBCE$ would require
$\tfrac{q(a)}2\rho_a\succeq\tfrac{q(a)}2\rho_b$ for both ordered pairs, i.e.\
$\rho_0\succeq\rho_1$ and $\rho_1\succeq\rho_0$. But
\[
\rho_0-\rho_1=\tfrac12\begin{pmatrix}1&-1\\-1&-1\end{pmatrix}
\qquad\text{has eigenvalues }\pm\tfrac1{\sqrt2},
\]
so neither domination holds. Quantitatively, the best deviation exceeds the obedience
bound by $\tfrac{q(a)}2\lambda_{\max}(\rho_0-\rho_1)=\tfrac{q(a)}{2\sqrt2}$;
equivalently the Helstrom success probability
$\tfrac12+\tfrac14\|\rho_0-\rho_1\|_1=\tfrac{2+\sqrt2}4\approx0.8536$ strictly exceeds
both the value $\tfrac12$ of ignoring the system and the value $\tfrac34$ of the
computational-basis measurement.
\end{proof}

\subsection{When the quantum structure carries no incentive content}\label{sec:trivial}

Before turning to the ordering we identify exactly which structures sit at its bottom.
Write $\rho_{i,\omega}:=\Tr_{-i}\rho_\omega$ for the local marginals, and call $Q$
\emph{locally flat} if $\rho_{i,\omega}$ is independent of $\omega$ for every $i$,
with common value $\rho_i$.

\begin{theorem}[Incentive-trivial structures]\label{thm:trivial}
$\QBCE(G,Q)=\BCE(G,S_\emptyset)$ for every basic game $G$ if and only if $Q$ is locally
flat.
\end{theorem}

\begin{proof}
($\Leftarrow$) Since $\QBCE\subseteq\QBCE_\pp$ and
$\QBCE_\pp(G,Q)=\BCE(G,S_\emptyset)$ by Theorem~\ref{thm:ppdegen}, only the inclusion
$\BCE(G,S_\emptyset)\subseteq\QBCE(G,Q)$ needs proof. Let
$\nu\in\BCE(G,S_\emptyset)$, so $\Xi_i(\nu;a_i,b_i)\ge0$ for all $i,a_i,b_i$, and take
the proportional rule $\rho_\omega^a=(\nu(\omega,a)/\psi(\omega))\rho_\omega$ of
Lemma~\ref{lem:realize}, whose outcome is $\nu$. By \eqref{eq:Xop} and local flatness,
\[
X_i^{a_i,b_i}
=\sum_{\omega,a_{-i}}\psi(\omega)u_i(b_i,a_{-i},\omega)\frac{\nu(\omega,a_i,a_{-i})}{\psi(\omega)}\,\rho_i
=\Big(\sum_{\omega,a_{-i}}u_i(b_i,a_{-i},\omega)\nu(\omega,a_i,a_{-i})\Big)\rho_i ,
\]
so $X_i^{a_i,a_i}-X_i^{a_i,b_i}=\Xi_i(\nu;a_i,b_i)\,\rho_i\succeq0$, and
Theorem~\ref{thm:loewner} applies.

($\Rightarrow$) Suppose $\rho_{i,\omega_0}\ne\rho_{i,\omega_1}$ for some player $i$
and states $\omega_0\ne\omega_1$. Take the basic game in which every player $j\ne i$
has a singleton action set, $A_i=\{0,1\}$, $u_j\equiv0$ for $j\ne i$,
\[
u_i(b_i,\omega):=\1\big[b_i=0,\ \omega=\omega_0\big]+\1\big[b_i=1,\ \omega=\omega_1\big],
\]
and $\psi$ any full-support prior with $\psi(\omega_0)=\psi(\omega_1)=:c$. The outcome
$\nu(\omega,b_i)=\tfrac12\psi(\omega)$ satisfies $\Xi_i\equiv0$, so it lies in
$\BCE(G,S_\emptyset)$. Let $\{\rho_\omega^{a}\}$ be any decision rule with this
outcome and put $R^{a}(\omega):=\Tr_{-i}[\rho_\omega^{a}]\succeq0$, so
$\sum_aR^a(\omega)=\rho_{i,\omega}$ and $\Tr R^a(\omega)=\tfrac12$ for all
$\omega,a$. Suppressing the singleton $a_{-i}$, \eqref{eq:Xop} gives
$X_i^{a,0}=c\,R^a(\omega_0)$ and $X_i^{a,1}=c\,R^a(\omega_1)$, so obedience
\eqref{eq:loewner} requires $R^0(\omega_0)\succeq R^0(\omega_1)$ at $a=0$ and
$R^1(\omega_1)\succeq R^1(\omega_0)$ at $a=1$. But
$\Tr[R^0(\omega_0)-R^0(\omega_1)]=\tfrac12-\tfrac12=0$, and a positive semidefinite
operator with zero trace vanishes, so $R^0(\omega_0)=R^0(\omega_1)$; symmetrically
$R^1(\omega_0)=R^1(\omega_1)$. Summing over $a$ gives
$\rho_{i,\omega_0}=\rho_{i,\omega_1}$, a contradiction. Hence no rule with this
outcome is obedient and $\QBCE(G,Q)\ne\BCE(G,S_\emptyset)$.
\end{proof}

All incentive content of a quantum information structure is therefore carried by the
$\omega$-dependence of its local marginals: obedience is an interim condition on what
player $i$ holds, and the only trace of the global state entering \eqref{eq:Xop} is
$\Tr_{-i}$ of the decision rule's components. When the marginals do not vary with
$\omega$, the proportional rule makes every obedience operator a scalar multiple of one
fixed $\rho_i$, collapsing the Loewner order to the scalar order. Restricted to
classical structures via Theorem~\ref{thm:classical}, the criterion says that $S$ is
BCE-equivalent to $S_\emptyset$ exactly when each player's marginal signal distribution
is $\omega$-independent, i.e.\ when every type of every player holds the prior belief
and this is common knowledge, which is the condition that $S$ has the canonical representation
of $S_\emptyset$, as in~\cite[Example~1]{BM2016}.

\section{Foundations of the solution concept}\label{sec:bm1}

Theorem~\ref{thm:bm1c} characterizes classical Bayes correlated equilibrium as the
behaviour arising in Bayes Nash equilibrium of an expansion. We
prove the quantum analogue, and then show that the mediator is not a physical operation
on the players' systems.

\begin{definition}[Expansion, split strategies]\label{def:expansion}
An \emph{expansion} of $Q$ is a structure $Q^\ast$ with $\Hil_i^\ast=\Hil_i\otimes E_i$
and $\Tr_E\rho_\omega^\ast=\rho_\omega$ for all $\omega$. Given an orthonormal basis
$\{\ket{e}\}$ of $E_i$, a \emph{split strategy} for player $i$ is the sequential
procedure: measure $E_i$ in that basis, obtain $e$, then measure $\Hil_i$ with a POVM
$\{N_e^{b_i}\}_{b_i\in A_i}$ and play the outcome. Because the two measurements act on
different tensor factors they commute, and the induced POVM on $\Hil_i^\ast$ has
elements $\sum_eN_e^{b_i}\otimes\proj e$.
\end{definition}

\begin{lemma}[Deviations collapse on a block-diagonal expansion]\label{lem:blockdiag}
Let $Q^\ast$ be an expansion whose ancillas carry orthonormal bases $\{\ket{a_i}\}$ and
suppose $\rho_\omega^\ast$ is block diagonal with respect to $\{\proj a\}_{a\in A}$,
$\proj a:=\bigotimes_i\proj{a_i}$, say
$\rho_\omega^\ast=\sum_a\varrho_\omega^a\otimes\proj a$. Fix $i$ and let every player
$j\ne i$ measure $E_j$ in $\{\ket{a_j}\}$ and play the outcome. Then every POVM
$\{M^{b_i}\}$ on $\Hil_i\otimes E_i$ induces the same joint distribution over
$(b_i,a_{-i},\omega)$ as the split strategy with second-stage POVMs
$M_{a_i}^{b_i}:=(\1\otimes\bra{a_i})M^{b_i}(\1\otimes\ket{a_i})$.
\end{lemma}

\begin{proof}
Write $M^{b_i}=\sum_{a_i,a_i'}M^{b_i}_{a_ia_i'}\otimes\ket{a_i}\!\bra{a_i'}$; the
diagonal blocks are the $M_{a_i}^{b_i}$. For each $a_i$ these form a POVM on $\Hil_i$:
a diagonal block of a positive semidefinite operator is positive semidefinite, and
$\sum_{b_i}M^{b_i}=\1$ gives $\sum_{b_i}M_{a_i}^{b_i}=\1_{\Hil_i}$ blockwise.

The probability of $(b_i,a_{-i})$ in state $\omega$ is
$\Tr\big[\big(M^{b_i}\otimes\bigotimes_{j\ne i}\proj{a_j}\big)\rho_\omega^\ast\big]$.
Expanding $\rho_\omega^\ast$, the factors
$\braket{a_j}{a_j'}\braket{a_j'}{a_j}$ force $a_{-i}'=a_{-i}$, while on
$\Hil_i\otimes E_i$,
\[
\Tr\big[M^{b_i}\big(X\otimes\proj{a_i'}\big)\big]
=\sum_{a_i,a_i''}\Tr\big[M^{b_i}_{a_ia_i''}X\big]\braket{a_i''}{a_i'}\braket{a_i'}{a_i}
=\Tr\big[M^{b_i}_{a_i'a_i'}X\big],
\]
so only the diagonal blocks contribute, giving the expression computed for the split
strategy.
\end{proof}

\begin{theorem}[Quantum analogue of BM Theorem 1]\label{thm:bm1}
Let $\{\rho_\omega^a\}$ be a quantum decision rule for $(G,Q)$, let $E_i:=\C^{A_i}$, and
let $Q^\ast$ be the \emph{canonical expansion}
$\rho_\omega^\ast:=\sum_{a\in A}\rho_\omega^a\otimes\proj a$. Then $Q^\ast$ is an
expansion of $Q$, and the \emph{truthful} profile, in which each player measures $E_i$ in the
basis $\{\ket{a_i}\}$ and plays the outcome, is a quantum Bayes Nash equilibrium of
$(G,Q^\ast)$ under split deviations if and only if $\{\rho_\omega^a\}$ is a QBCE of
$(G,Q)$. The truthful profile induces $\{\rho_\omega^a\}$.
\end{theorem}

\begin{proof}
$\Tr_E\rho_\omega^\ast=\sum_a\rho_\omega^a=\rho_\omega$, so $Q^\ast$ is an expansion.
Suppose all players $j\ne i$ play truthfully and player $i$ uses the split strategy
with second-stage POVMs $\{M_{a_i}^{b_i}\}$. The probability of the joint action
profile $(b_i,a_{-i})$ in state $\omega$ is
\[
\sum_{a'\in A}\Tr\Big[\Big(\big(\textstyle\sum_{a_i}M_{a_i}^{b_i}\otimes\proj{a_i}\big)
\otimes\bigotimes_{j\ne i}\proj{a_j}\Big)\big(\rho_\omega^{a'}\otimes\proj{a'}\big)\Big]
\]
\[
\;=\;\sum_{a_i}\Tr\big[(M_{a_i}^{b_i}\otimes\1_{-i})\rho_\omega^{(a_i,a_{-i})}\big],
\]
since the ancilla projectors force $a'=(a_i,a_{-i})$. Player $i$'s expected payoff is
therefore
\[
\sum_{\omega,b_i,a_{-i}}\psi(\omega)u_i(b_i,a_{-i},\omega)\sum_{a_i}\Tr\big[(M_{a_i}^{b_i}\otimes\1)\rho_\omega^{(a_i,a_{-i})}\big]
=\sum_{a_i}\sum_{b_i}\Tr\big[M_{a_i}^{b_i}X_i^{a_i,b_i}\big]
\]
by \eqref{eq:Xop}, and the truthful choice $M_{a_i}^{b_i}=\delta_{a_ib_i}\1$ gives
$\sum_{a_i}\Tr[X_i^{a_i,a_i}]$. The constraint decouples over $a_i$, the families
$\{M_{a_i}^{b_i}\}_{b_i}$ being independent POVMs, so truthfulness is optimal if and
only if for each $a_i$ and every POVM $\{M^{b_i}\}$ on $\Hil_i$,
$\sum_{b_i}\Tr[M^{b_i}X_i^{a_i,b_i}]\le\Tr[X_i^{a_i,a_i}]$, which is
Definition~\ref{def:qbce}. Taking $M_{a_i}^{b_i}=\delta_{a_ib_i}\1$ in the displayed
probability shows the truthful profile induces $\{\rho_\omega^a\}$.
\end{proof}

\begin{corollary}[Arbitrary deviations]\label{cor:bm1general}
In Theorem~\ref{thm:bm1} the deviation class may be enlarged to \emph{all} POVMs on
$\Hil_i\otimes E_i$, including those entangling the recommendation register with the
payoff-relevant subsystem.
\end{corollary}

\begin{proof}
The canonical expansion state is block diagonal with respect to $\{\proj a\}$, so by
Lemma~\ref{lem:blockdiag} every POVM on $\Hil_i\otimes E_i$ is payoff-equivalent to a
split strategy. The supremum of player $i$'s payoff over arbitrary deviations therefore
equals the supremum over split deviations, and Theorem~\ref{thm:bm1} applies.
\end{proof}

A block-diagonal state is blind to the off-diagonal blocks of a measurement operator:
coherence between distinct recommendations is simply not present in
$\rho_\omega^\ast$, so no deviation can exploit it. Measurement disturbance, which
might be expected to obstruct the argument, does not.

\begin{proposition}[The mediator is not a quantum instrument]\label{prop:noinstrument}
There are a quantum information structure $Q$ with $\Ntot=1$, a decision problem $D$,
and a QBCE decision rule of $(D,Q)$ for which no family of positive maps, and in particular no
family of completely positive maps, $\mathcal{I}_a$ satisfies
$\mathcal{I}_a(\rho_\omega)=\rho_\omega^a$ for all $\omega$ and $a$.
\end{proposition}

\begin{proof}
Let $\Hil=\C^2$, $\Omega=\{1,2,3,4\}$, $\psi$ uniform, and $\rho_1=\proj0$,
$\rho_2=\proj1$, $\rho_3=\proj+$, $\rho_4=\proj{{+}i}$ with
$\ket{{+}i}=(\ket0+i\ket1)/\sqrt2$; their Bloch vectors $e_z,-e_z,e_x,e_y$ affinely
span $\R^3$, so $\{\rho_\omega\}$ spans the real space of Hermitian operators on
$\C^2$. Let $A=\{0,1\}$ and take
$\rho_\omega^0:=\1[\omega=1]\rho_\omega$, $\rho_\omega^1:=\1[\omega\ne1]\rho_\omega$, a
decision rule. It is a QBCE for the decision problem with
$u(a,\omega):=\1[a=0,\omega=1]+\1[a=1,\omega\ne1]$: by Theorem~\ref{thm:loewner},
obedience at $a=0$ requires
$\sum_\omega\psi(\omega)[u(0,\omega)-u(1,\omega)]\rho_\omega^0=\psi(1)\rho_1\succeq0$
and at $a=1$ requires
$\sum_{\omega\ne1}\psi(\omega)\rho_\omega\succeq0$, both of which hold.

Suppose $\mathcal{I}_0$ were linear with $\mathcal{I}_0(\rho_\omega)=\rho_\omega^0$ for
all $\omega$. Since the $\rho_\omega$ span the Hermitian operators, $\mathcal{I}_0$ is
uniquely determined; solving in the Pauli basis from $\mathcal{I}_0(\proj0)=\proj0$ and
$\mathcal{I}_0(\rho_\omega)=0$ for $\omega\ne1$ gives
$\mathcal{I}_0(\1)=\mathcal{I}_0(\sigma_z)=\proj0$ and
$\mathcal{I}_0(\sigma_x)=\mathcal{I}_0(\sigma_y)=-\proj0$, so for
$\rho(r)=\tfrac12(\1+r\cdot\vec\tau)$,
\[
\mathcal{I}_0\big(\rho(r)\big)=\tfrac12\big(1+r_z-r_x-r_y\big)\proj0 .
\]
At the unit Bloch vector $r=(1,1,-1)/\sqrt3$ the coefficient is $1-\sqrt3<0$, so
$\mathcal{I}_0(\rho(r))$ is not positive semidefinite. Hence $\mathcal{I}_0$ is not
positive, and a fortiori not completely positive.
\end{proof}

Proposition~\ref{prop:noinstrument} is not a technical obstruction to be worked around;
it records that the mediator is not a physical device, and that this is deliberate. An
instrument is an operation applied to the players' systems by someone who does not know
$\omega$. The mediator does know $\omega$, and splits the ensemble conditionally on it.
A rule such as ``recommend $0$ exactly when $\omega=1$'' is therefore available to the
mediator and unavailable to any channel, since no physical map can read $\omega$ off
states it cannot distinguish.

This is what the omniscience in Definition~\ref{def:qdr} buys, and
Theorem~\ref{thm:bm1} is what makes it legitimate. The equilibria are exactly the
behaviours realizable in an expansion whose state $\rho_\omega^\ast$ is prepared
conditionally on $\omega$, and preparing a state conditionally on a classical parameter
is an ordinary physical operation. What is not physical is the attempt to compress that
preparation into a single $\omega$-independent operation on $Q$, and
Proposition~\ref{prop:noinstrument} shows that no such compression exists.

\section{Quantum individual sufficiency and the comparison theorem}\label{sec:order}

We now build the order on quantum information structures and prove the comparison
theorem. The classical template is Definition~\ref{def:is}, whose two operative
features are a transfer of decision rules that may depend on $\omega$ and a single
local kernel per player depending on neither $\omega$ nor the other players' data. One
question has to be settled before the quantum definition can be written down, namely
which class of maps the transfer should range over. The obvious candidate is the class
of channels, since those are the physically implementable transformations and the
single-agent theory of~\cite{Buscemi2012} is built from them. The following theorem
shows that this candidate cannot work.

\begin{theorem}[Transposition]\label{thm:transpose}
Let $Q=\big((\Hil_i),(\rho_\omega)\big)$ be a quantum information structure and let
$Q^{\mathsf T}:=\big((\Hil_i),(\rho_\omega^{\mathsf T})\big)$, the transpose being taken
in a fixed product basis of $\bigotimes_i\Hil_i$. Then
\begin{enumerate}[label=(\roman*),leftmargin=2em]
  \item $\QBCE(G,Q)=\QBCE(G,Q^{\mathsf T})$ for every basic game $G$ and every number of
  players, so $Q$ and $Q^{\mathsf T}$ are indistinguishable by the incentive ordering;
  \item there are structures $Q$ for which no channel carries the ensemble of $Q$ to
  that of $Q^{\mathsf T}$, nor conversely.
\end{enumerate}
Consequently no relation defined through completely positive maps can characterize the
incentive ordering, even for a single player.
\end{theorem}

\begin{proof}
(i) Transposition in a product basis commutes with every partial trace, that is,
$\Tr_{-i}[\xi^{\mathsf T}]=(\Tr_{-i}\xi)^{\mathsf T}$: by linearity it suffices to check
this on product operators $\xi_i\otimes\xi_{-i}$, where both sides equal
$\xi_i^{\mathsf T}\Tr[\xi_{-i}]$.

Let $\{\rho_\omega^a\}$ be a decision rule for $Q$. Transposition is linear and
preserves positive semidefiniteness, so the operators $(\rho_\omega^a)^{\mathsf T}$ are
positive semidefinite and sum to $\rho_\omega^{\mathsf T}$, making
$\{(\rho_\omega^a)^{\mathsf T}\}$ a decision rule for $Q^{\mathsf T}$; and it preserves
traces, so the induced outcome is unchanged. By the commutation just proved and
\eqref{eq:Xop}, the obedience operators of the transposed rule are
$(X_i^{a_i,b_i})^{\mathsf T}$. Transposition preserves the Loewner order, since
$Z\succeq0$ if and only if $Z^{\mathsf T}\succeq0$, so by Theorem~\ref{thm:loewner} the
transposed rule is a quantum Bayes correlated equilibrium exactly when the original one
is. The assignment is an involution, so the two outcome sets coincide.

(ii) Take $\Ntot=1$, $\Hil=\C^2$, $\Omega=\{1,2,3,4\}$ and
\[
\rho_1=\proj0,\quad\rho_2=\proj1,\quad\rho_3=\proj+,\quad\rho_4=\proj{{+}i},
\qquad \ket{{+}i}:=\tfrac1{\sqrt2}(\ket0+i\ket1).
\]
Let $e_x,e_y,e_z$ be the standard basis of $\R^3$ and $\vec\tau$ the vector of Pauli
matrices, and write a qubit state as $\rho(r)=\tfrac12(\1+r\cdot\vec\tau)$ for
$r\in\R^3$ with $\|r\|\le1$. Every channel on a qubit acts on Bloch vectors as an
affine map $r\mapsto\Theta r+c$. The Bloch vectors of $\rho_1,\dots,\rho_4$ are
$e_z,-e_z,e_x,e_y$, whose affine hull is all of $\R^3$, so $\Theta$ and $c$ are
determined by the four images. Transposition is complex conjugation in the
computational basis and acts as $(x,y,z)\mapsto(x,-y,z)$, so the images are
$e_z,-e_z,e_x,-e_y$. From the images of $\pm e_z$ we get $\Theta e_z+c=e_z$ and
$-\Theta e_z+c=-e_z$, whence $c=0$ and $\Theta e_z=e_z$; then $\Theta e_x=e_x$ and
$\Theta e_y=-e_y$, so $\Theta=\operatorname{diag}(1,-1,1)$ and the map is transposition
itself. Its Choi operator is
$\sum_{jk}\ket j\!\bra k\otimes(\ket j\!\bra k)^{\mathsf T}$, the swap operator, whose
spectrum is $\{-1,1,1,1\}$; it is not positive semidefinite, so transposition is not
completely positive and no channel realizes the required action. Exchanging the roles of
$Q$ and $Q^{\mathsf T}$ gives the other direction.

The final claim follows: by (i) any relation characterizing the incentive ordering must
relate $Q$ and $Q^{\mathsf T}$ in both directions, and by (ii) no relation defined
through channels does.
\end{proof}

The mechanism is that QBCE is defined with \emph{classical} actions: for any POVM
$\{M^a\}$ the family $\{(M^a)^{\mathsf T}\}$ is again a POVM, and
$\Tr[M^a\rho_\omega^{\mathsf T}]=\Tr[(M^a)^{\mathsf T}\rho_\omega]$, so no
classical-action test can distinguish an ensemble from its transpose. Buscemi's
characterization of completely positive sufficiency~\cite{Buscemi2012} uses tests on
$\rho_\omega$ entangled with a reference system, which are exactly the tests this
framework excludes by construction (Section~\ref{sub:classicalactions}). The incentive
ordering is therefore strictly coarser than the order of~\cite{Buscemi2012}.

Transposition is nevertheless positive and trace preserving, and this is the property
the definition should retain. We therefore quantize Definition~\ref{def:is} using
positive trace-preserving maps rather than channels.

\begin{definition}[Quantum individual sufficiency]\label{def:qis-order}
$Q\qis Q'$ if there exist
\begin{enumerate}[label=(\alph*),leftmargin=1.8em]
  \item positive trace-preserving maps $\mathcal{T}_\omega:\Bop(\Hil)\to\Bop(\Hil')$ with
  $\mathcal{T}_\omega(\rho_\omega)=\rho'_\omega$ for every $\omega$, and
  \item for each $i$ a positive trace-preserving map
  $\Phi_i:\Bop(\Hil_i)\to\Bop(\Hil_i')$, \emph{independent of $\omega$},
\end{enumerate}
such that the \emph{local compatibility} identity
\begin{equation}\label{eq:loccomp}
\Tr_{-i}\big[\mathcal{T}_\omega(\xi)\big]\;=\;\Phi_i\big(\Tr_{-i}[\xi]\big)
\end{equation}
holds for every $\omega$, every $i$, and every $\xi\in\Bop(\Hil)$ supported on
$\supp\rho_\omega$. We write $Q\inc Q'$ for the \emph{incentive ordering},
$\QBCE(G,Q)\subseteq\QBCE(G,Q')$ for every basic game $G$.
\end{definition}

The $\omega$-independence of $\Phi_i$ is the exact counterpart of BM's requirement that
$\Pr(t_i'\mid t_i,t_{-i},\omega)$ depend on neither $\omega$ nor $t_{-i}$; the
$\omega$-dependence permitted in $\mathcal{T}_\omega$ mirrors the fact that BM's transfer
\cite[eq.~(8)]{BM2016} is built from $\pi^\ast(\cdot\mid\omega)$. Restricting
\eqref{eq:loccomp} to $\supp\rho_\omega$ is the counterpart of BM's ``whenever the
denominator is nonzero'', and is harmless because every decision rule satisfies
$\rho_\omega^a\preceq\rho_\omega$.

\begin{proposition}[$\qis$ and $\inc$ are preorders]\label{prop:preorder}
Both relations are reflexive and transitive.
\end{proposition}
\begin{proof}
$\inc$ is a preorder because set inclusion is. For $\qis$, reflexivity holds with
$\mathcal{T}_\omega=\mathrm{id}$ and $\Phi_i=\mathrm{id}$. For transitivity let
$(\mathcal{T}_\omega,\Phi_i)$ witness $Q\qis Q'$ and $(\mathcal{T}'_\omega,\Phi'_i)$ witness $Q'\qis Q''$,
and set $\mathcal{T}''_\omega:=\mathcal{T}'_\omega\circ \mathcal{T}_\omega$, $\Phi''_i:=\Phi'_i\circ\Phi_i$.
Compositions of positive maps are positive and compositions of trace-preserving maps
are trace preserving; $\mathcal{T}''_\omega(\rho_\omega)=\mathcal{T}'_\omega(\rho'_\omega)=\rho''_\omega$;
and each $\Phi''_i$ is $\omega$-independent.

For \eqref{eq:loccomp} we need that $\mathcal{T}_\omega$ maps operators supported on
$\supp\rho_\omega$ to operators supported on $\supp\rho'_\omega$. Let $\xi\succeq0$
with $\supp\xi\subseteq\supp\rho_\omega$; then $\xi\preceq c\rho_\omega$ for some
$c>0$, so positivity and linearity give $\mathcal{T}_\omega(c\rho_\omega-\xi)\succeq0$, i.e.\
$0\preceq \mathcal{T}_\omega(\xi)\preceq c\rho'_\omega$, whence
$\supp \mathcal{T}_\omega(\xi)\subseteq\supp\rho'_\omega$. A general Hermitian $\xi$ supported on
$\supp\rho_\omega$ is a difference of two such positive operators, so the conclusion
extends by linearity. Therefore
\[
\Tr_{-i}\big[\mathcal{T}''_\omega(\xi)\big]=\Phi'_i\big(\Tr_{-i}[\mathcal{T}_\omega\xi]\big)
=\Phi'_i\big(\Phi_i(\Tr_{-i}[\xi])\big)=\Phi''_i\big(\Tr_{-i}[\xi]\big).
\]
\end{proof}

\begin{theorem}[Soundness: more information shrinks the equilibrium set]\label{thm:soundness}
If $Q\qis Q'$ then $Q\inc Q'$. That is, for every basic game $G$ on the common basic
structure and every number of players,
\[
\QBCE(G,Q)\ \subseteq\ \QBCE(G,Q').
\]
\end{theorem}
\begin{proof}
Let $(\mathcal{T}_\omega),(\Phi_i)$ witness $Q\qis Q'$, let $\nu\in\QBCE(G,Q)$ be realized by a
QBCE decision rule $\{\rho_\omega^a\}$, and set $\rho_\omega'^{\,a}:=\mathcal{T}_\omega(\rho_\omega^a)$.

\emph{Step 1: a decision rule for $Q'$ with outcome $\nu$.} Positivity gives
$\rho_\omega'^{\,a}\succeq0$; linearity gives
$\sum_a\rho_\omega'^{\,a}=\mathcal{T}_\omega(\rho_\omega)=\rho'_\omega$; and trace preservation
gives $\Tr[\rho_\omega'^{\,a}]=\Tr[\rho_\omega^a]$, so the outcome is again $\nu$.

\emph{Step 2: obedience operators transform by $\Phi_i$.} Since
$\rho_\omega^a\preceq\rho_\omega$, identity \eqref{eq:loccomp} applies with
$\xi=\rho_\omega^a$ and yields
$\Tr_{-i}[\rho_\omega'^{\,a}]=\Phi_i\big(\Tr_{-i}[\rho_\omega^a]\big)$. Substituting
into \eqref{eq:Xop} and using linearity of $\Phi_i$ together with the fact that the
coefficients $\psi(\omega)u_i(b_i,a_{-i},\omega)$ are real scalars,
\[
X_i'^{\,a_i,b_i}
=\sum_{\omega,a_{-i}}\psi(\omega)u_i(b_i,a_{-i},\omega)\,\Phi_i\big(\Tr_{-i}[\rho_\omega^{(a_i,a_{-i})}]\big)
=\Phi_i\big(X_i^{a_i,b_i}\big).
\]

\emph{Step 3: obedience transfers.} By Theorem~\ref{thm:loewner},
$X_i^{a_i,a_i}-X_i^{a_i,b_i}\succeq0$ for all $i,a_i,b_i$. Since $\Phi_i$ is positive
and linear,
$X_i'^{\,a_i,a_i}-X_i'^{\,a_i,b_i}=\Phi_i\big(X_i^{a_i,a_i}-X_i^{a_i,b_i}\big)\succeq0$.
By Theorem~\ref{thm:loewner} again, $\{\rho_\omega'^{\,a}\}$ is a QBCE of $(G,Q')$ with
outcome $\nu$, so $\nu\in\QBCE(G,Q')$.
\end{proof}

It is worth reading the proof back to see which hypotheses it consumed, because the
answer determines the shape of the theory.

Step 1 succeeded because realizability is free (Lemma~\ref{lem:realize}). The coarser
structure can still \emph{produce} the outcome, since the mediator conditions on
$\omega$ rather than on the degraded quantum data, so nothing is lost in the transfer
that would have to be recovered. This is where an alternative definition would fail. If
recommendations were instead generated by measuring $\rho_\omega$, as under a
belief-invariant reading, then a garbled structure would in general be unable to
reproduce the outcome at all, because data processing strictly reduces
distinguishability, and no soundness theorem of the present form would be available.
The omniscience of the mediator is thus not a convenience but the hypothesis that makes
the comparison possible.

Steps 2 and 3 consumed only two properties of the maps: linearity, which carries the
transfer through the definition \eqref{eq:Xop} of the obedience operators, and
positivity, which preserves the Loewner order. Nothing else about
$\mathcal{T}_\omega$ or $\Phi_i$ was used. Two consequences follow.

The first is that the hypothesis on $\mathcal{T}_\omega$ can be weakened. Step 1
applies it only to the operators $\rho_\omega^a$, which lie in the order interval
$[0,\rho_\omega]$, so positivity on that interval suffices and
Theorem~\ref{thm:soundness} remains valid under the weaker requirement. The same
weakening is not available for $\Phi_i$, since Step 3 applies it to the differences
$X_i^{a_i,a_i}-X_i^{a_i,b_i}$, which are not confined to any such interval. The
asymmetry is structural: $\mathcal{T}_\omega$ acts on the mediator's splitting of a
fixed state, whereas $\Phi_i$ acts on incentive differences.

The second concerns complete positivity, which the proof also never used. One may
therefore strengthen the hypothesis and require the maps of
Definition~\ref{def:qis-order} to be channels. Call the resulting relation
$\succeq_{\mathrm{CP}}$. Since every channel is a positive trace-preserving map,
$\succeq_{\mathrm{CP}}\subseteq\qis$, so Theorem~\ref{thm:soundness} applies to it
verbatim; and every witness constructed in this paper happens to be completely
positive, so nothing proved here would be lost by adopting it. What would be lost is
the converse. By Theorem~\ref{thm:transpose} the incentive ordering identifies an
ensemble with its transpose, and no channel relates the two, so $\succeq_{\mathrm{CP}}$
cannot characterize the incentive ordering. This is why
Definition~\ref{def:qis-order} is stated with positive maps and not with channels, which are completely positive trace-preserving maps.

\subsection{The order is inhabited}

\begin{definition}[Three garbling relations]\label{def:garb}
Let $Q,Q'$ be quantum information structures.
\begin{enumerate}[label=(\alph*),leftmargin=1.8em]
  \item $Q\garb Q'$ (\emph{local garbling}) if there are channels
  $\Lambda_i:\Bop(\Hil_i)\to\Bop(\Hil_i')$ with
  $\rho'_\omega=\big(\bigotimes_i\Lambda_i\big)(\rho_\omega)$ for all $\omega$.
  \item $Q\garbsr Q'$ (\emph{shared randomness}) if $\Hil_i'=\Hil_i''\otimes\C^{\mathcal{L}}$
  and there are $p\in\Delta(\mathcal{L})$ and channels $\Lambda_i^\lambda$ with
  $\rho'_\omega=\sum_\lambda p(\lambda)\big(\bigotimes_i\Lambda_i^\lambda\big)(\rho_\omega)\otimes\big(\bigotimes_i\proj\lambda\big)$.
  \item $Q\garbe Q'$ (\emph{shared entanglement}) if there are auxiliary spaces $K_i$,
  a state $\tau$ on $\bigotimes_iK_i$ independent of $\omega$, and channels
  $\Lambda_i:\Bop(\Hil_i\otimes K_i)\to\Bop(\Hil_i')$ with
  $\rho'_\omega=\big(\bigotimes_i\Lambda_i\big)(\rho_\omega\otimes\tau)$.
\end{enumerate}
\end{definition}

\begin{proposition}[Instances]\label{prop:instances}
Each of $\garb$, $\garbsr$ and $\garbe$ implies $\qis$, with $\omega$-independent
witnesses $\mathcal{T}_\omega=T$ and
\[
\text{(a) } \mathcal{T}=\bigotimes_i\Lambda_i,\ \Phi_i=\Lambda_i;\quad
\text{(b) } \Phi_i(Z)=\sum_\lambda p(\lambda)\Lambda_i^\lambda(Z)\otimes\proj\lambda;\quad
\text{(c) } \Phi_i(Z)=\Lambda_i(Z\otimes\tau_i),
\]
where $\tau_i:=\Tr_{K_{-i}}\tau$. In each case $\Phi_i$ is a channel.
\end{proposition}

\begin{proof}
We use the identity
\begin{equation}\label{eq:tracecommute}
\Tr_{-i}\Big[\Big(\bigotimes_j\Gamma_j\Big)(\xi)\Big]=\Gamma_i\big(\Tr_{-i}[\xi]\big)
\end{equation}
for channels $\Gamma_j$ acting on the respective factors: by linearity it suffices to
check it on product operators $\xi=\xi_i\otimes\xi_{-i}$, where the left side is
$\Gamma_i(\xi_i)\Tr[\Gamma_{-i}(\xi_{-i})]=\Gamma_i(\xi_i)\Tr[\xi_{-i}]$ because
$\Gamma_{-i}:=\bigotimes_{j\ne i}\Gamma_j$ is trace preserving.

(a) Take $\mathcal{T}=\bigotimes_i\Lambda_i$; then \eqref{eq:loccomp} is
\eqref{eq:tracecommute}. 

(b) Set
$\mathcal{T}(\xi):=\sum_\lambda p(\lambda)\big(\bigotimes_i\Lambda_i^\lambda\big)(\xi)\otimes\big(\bigotimes_i\proj\lambda\big)$,
completely positive as a convex combination of completely positive maps and trace
preserving since each summand contributes $p(\lambda)\Tr\xi$. Tracing out the factors
$j\ne i$ and applying \eqref{eq:tracecommute} within each $\lambda$-term gives
$\Tr_{-i}[\mathcal{T}(\xi)]=\sum_\lambda p(\lambda)\Lambda_i^\lambda(\Tr_{-i}[\xi])\otimes\proj\lambda$,
and $\Tr\Phi_i(Z)=\sum_\lambda p(\lambda)\Tr Z=\Tr Z$. 

(c) Set
$\mathcal{T}(\xi):=\big(\bigotimes_i\Lambda_i\big)(\xi\otimes\tau)$, a channel. Player $j$ holds
$\Hil_j\otimes K_j$, so by \eqref{eq:tracecommute}
$\Tr_{-i}[\mathcal{T}(\xi)]=\Lambda_i\big(\Tr_{(\Hil\otimes K)_{-i}}[\xi\otimes\tau]\big)$; since
$\xi$ lives on $\Hil$ and $\tau$ on $K$, the partial trace factorizes as
$\Tr_{-i}[\xi]\otimes\tau_i$, giving $\Phi_i$ as displayed, with
$\Tr\Phi_i(Z)=\Tr[Z]\Tr[\tau_i]=\Tr Z$.
\end{proof}

\begin{proposition}[Classical individual sufficiency is an instance]\label{prop:classicalIS}
If $S\suff S'$ in the sense of Definition~\ref{def:is}, then $\iota(S)\qis\iota(S')$.
\end{proposition}

\begin{proof}
Let $\pi^\ast$ witness $S\suff S'$ and let $\varphi_i$ satisfy \eqref{eq:phi}. For each
$\omega$ define the Markov kernel
$k_\omega(t'\mid t):=\pi^\ast(t,t'\mid\omega)/\pi(t\mid\omega)$ when
$\pi(t\mid\omega)>0$, and arbitrarily otherwise; let $\mathcal{T}_\omega$ be the induced channel
on diagonal operators, extended to $\Bop(\Hil)$ by first pinching in the product basis.
Then
$\mathcal{T}_\omega(\rho_\omega)=\sum_{t'}\big(\sum_t\pi(t\mid\omega)k_\omega(t'\mid t)\big)\proj{t'}
=\sum_{t'}\pi'(t'\mid\omega)\proj{t'}=\rho'_\omega$, using that $\pi^\ast$ has marginal
$\pi'$.

Let $\Phi_i$ be the channel induced in the same way by $\varphi_i$; it is trace
preserving because $\sum_{t_i'}\varphi_i(t_i'\mid t_i)=1$, positive because
$\varphi_i\ge0$, and independent of $\omega$. Let $\xi$ be supported on
$\supp\rho_\omega$; after pinching we may write $\xi=\sum_tc_t\proj t$ with $c_t=0$
unless $\pi(t\mid\omega)>0$. For such $t$, equation \eqref{eq:phi} gives
$\sum_{t_{-i}'}k_\omega(t'\mid t)=\varphi_i(t_i'\mid t_i)$, which depends on neither
$t_{-i}$ nor $\omega$; hence
\[
\Tr_{-i}\big[\mathcal{T}_\omega(\xi)\big]
=\sum_{t_i'}\Big(\sum_{t_i}\Big(\sum_{t_{-i}}c_t\Big)\varphi_i(t_i'\mid t_i)\Big)\proj{t_i'}
=\Phi_i\big(\Tr_{-i}[\xi]\big).
\]
\end{proof}

\begin{corollary}[Consistency with the classical and single-agent theories]\label{cor:consistency}
\begin{enumerate}[label=(\roman*),leftmargin=*]
  \item If $S\suff S'$ then \\$\BCE(G,S)\subseteq\BCE(G,S')$ for every basic game $G$;
  the soundness half of Theorem~\ref{thm:bm} is recovered.
  \item For $\Ntot=1$, if $Q$ is Buscemi sufficient for $Q'$, that is, if there is a channel
  $\Lambda$ with $\rho'_\omega=\Lambda(\rho_\omega)$ for all $\omega$, then
  $\QBCE(D,Q)\subseteq\QBCE(D,Q')$ for every decision problem $D$.
\end{enumerate}
\end{corollary}

\begin{proof}
(i) By Proposition~\ref{prop:classicalIS}, $\iota(S)\qis\iota(S')$; by
Theorem~\ref{thm:soundness}, $\QBCE(G,\iota(S))\subseteq\QBCE(G,\iota(S'))$; by
Theorem~\ref{thm:classical} these sets are $\BCE(G,S)$ and $\BCE(G,S')$. (ii) Buscemi
sufficiency is $Q\garb Q'$ for $\Ntot=1$; apply
Proposition~\ref{prop:instances}(a) and Theorem~\ref{thm:soundness}.
\end{proof}

\begin{proposition}[Information does not move the top of the equilibrium set]\label{prop:notvalue}
Let $\Ntot=1$, so that the basic game $G$ reduces to a single-agent decision problem $D$. Then for every quantum information structure $Q,$ the maximum payoff over the quantum Bayes correlated equilibrium set satisfies
\[
\max_{\nu\in\QBCE(D,Q)}\ \sum_{\omega,a}u(a,\omega)\,\nu(\omega,a)
\;=\;\sum_\omega\psi(\omega)\max_{a\in A}u(a,\omega)
\]
That is, the maximum attainable payoff
equals the full-information optimum and does not depend on $Q$.
\end{proposition}

\begin{proof}
No outcome can exceed the stated right-hand side, since

\begin{equation}
\sum_{\omega,a}u(a,\omega)\nu(\omega,a)\le\sum_\omega\big(\max_au(a,\omega)\big)\sum_a\nu(\omega,a)
=\sum_\omega\psi(\omega)\max_au(a,\omega).
\end{equation}
For attainment, fix a selector $a^\ast(\omega)\in\arg\max_au(a,\omega)$ and take the
decision rule $\rho_\omega^a:=\1[a=a^\ast(\omega)]\,\rho_\omega$, which is a decision
rule by Lemma~\ref{lem:realize} and has outcome
$\nu(\omega,a)=\psi(\omega)\1[a=a^\ast(\omega)]$. With one player the sum over $a_{-i}$
in \eqref{eq:Xop} is empty, so
$X^{a,b}=\sum_\omega\psi(\omega)u(b,\omega)\1[a=a^\ast(\omega)]\rho_\omega$ and
\[
X^{a,a}-X^{a,b}
=\sum_{\omega\,:\,a^\ast(\omega)=a}\psi(\omega)\big[u(a,\omega)-u(b,\omega)\big]\rho_\omega .
\]
Each summand is a nonnegative multiple of the positive semidefinite operator
$\rho_\omega$, because $a=a^\ast(\omega)$ maximizes $u(\cdot,\omega)$ on the index set
of the sum. Hence the difference is positive semidefinite and the rule is a quantum
Bayes correlated equilibrium by Theorem~\ref{thm:loewner}. Its payoff is
$\sum_\omega\psi(\omega)u(a^\ast(\omega),\omega)$, the right-hand side.
\end{proof}

Information therefore does not move the top of the equilibrium set; it deletes elements
elsewhere. With several players and a welfare objective the situation differs, because
the welfare-maximizing profile need not be obedient, and Section~\ref{sec:example}
exhibits that difference.

\subsection{Completeness}\label{sec:complete}

Theorem~\ref{thm:soundness} states that $\qis$ is contained in $\inc$ as a relation on
quantum information structures. In the classical theory the two corresponding relations
are equal~\cite[Thm.~2]{BM2016}. We establish equality here on two classes of pairs.

\begin{theorem}[Classical completeness]\label{thm:classicalcomplete}
Let $Q,Q'$ have ensembles simultaneously diagonalizable in product bases, i.e.\
$Q=(\bigotimes_iU_i)\iota(S)(\bigotimes_iU_i)^\dagger$ and
$Q'=(\bigotimes_iU_i')\iota(S')(\bigotimes_iU_i')^\dagger$ for classical structures
$S,S'$ and local unitaries. Then
\[
Q\inc Q'\quad\Longleftrightarrow\quad Q\qis Q'\quad\Longleftrightarrow\quad S\suff S' ,
\]
so the two relations agree on such pairs at every $\Ntot$.
\end{theorem}

\begin{proof}
Conjugation by a profile of local unitaries is a local garbling in both directions
(Definition~\ref{def:garb}(a) with $\Lambda_i=U_i\cdot U_i^\dagger$), so by
Proposition~\ref{prop:instances} and Theorem~\ref{thm:soundness} it preserves both
$\qis$-comparability and QBCE outcome sets; we may therefore assume $Q=\iota(S)$ and
$Q'=\iota(S')$. Now $S\suff S'\Rightarrow Q\qis Q'$ is
Proposition~\ref{prop:classicalIS} and $Q\qis Q'\Rightarrow Q\inc Q'$ is
Theorem~\ref{thm:soundness}. Conversely, if $Q\inc Q'$ then by
Theorem~\ref{thm:classical} $\BCE(G,S)\subseteq\BCE(G,S')$ for every basic game $G$,
i.e.\ $S$ is more incentive constrained than $S'$, so $S\suff S'$ by
Theorem~\ref{thm:bm}.
\end{proof}

What is left unresolved therefore concerns only ensembles that are not simultaneously
diagonalizable in a product basis. Note that the reduction to BM's theorem is
legitimate only because the embedding $\iota$ is \emph{exact}
(Theorem~\ref{thm:classical}); an inclusion would not suffice.

\begin{theorem}[Completeness at the bottom of the order]\label{thm:bottomcomplete}
Let $Q_\emptyset$ denote the trivial structure, $\Hil_i=\C$ and $\rho_\omega=1$ for all
$\omega$. Then for every quantum information structure $Q$,
\[
Q\qis Q_\emptyset \iff Q\inc Q_\emptyset \qquad(\text{both always hold}),
\]
and
\[
Q_\emptyset\qis Q \iff Q_\emptyset\inc Q \iff Q\ \text{is locally flat}.
\]
In particular the two relations agree whenever one of the two structures is
$Q_\emptyset$.
\end{theorem}

\begin{proof}
$Q\qis Q_\emptyset$ always, via the replacer $\mathcal{T}_\omega(\xi):=\Tr[\xi]$ and
$\Phi_i(X):=\Tr[X]$, which are positive and trace preserving and satisfy
\eqref{eq:loccomp}; $Q\inc Q_\emptyset$ then follows from
Theorem~\ref{thm:soundness}, and also directly from
$\QBCE(G,Q)\subseteq\QBCE_\pp(G,Q)=\BCE(G,S_\emptyset)=\QBCE(G,Q_\emptyset)$, the last
equality by Theorem~\ref{thm:classical} applied to $S_\emptyset$.

For the second chain, note first that $\Hil=\C$ for $Q_\emptyset$, so $\Tr_{-i}$ is the
identity on $\Bop(\C)=\C$ and any $\mathcal{T}_\omega$ with $\mathcal{T}_\omega(1)=\rho'_\omega$ acts by
$\mathcal{T}_\omega(\xi)=\xi\rho'_\omega$. Condition \eqref{eq:loccomp} then reads
$\xi\,\rho'_{i,\omega}=\Phi_i(\xi)=\xi\,\Phi_i(1)$ for all scalars $\xi$, i.e.\
$\rho'_{i,\omega}=\Phi_i(1)$ for every $\omega$. Such an $\omega$-independent $\Phi_i$
exists if and only if $\rho'_{i,\omega}$ does not depend on $\omega$, that is, if and
only if $Q$ is locally flat; and in that case $\Phi_i(\xi):=\xi\rho'_i$ is positive
and trace preserving. Hence $Q_\emptyset\qis Q$ if and only if $Q$ is locally flat.

Finally $Q_\emptyset\inc Q$ means $\BCE(G,S_\emptyset)\subseteq\QBCE(G,Q)$ for every
$G$; combined with the reverse inclusion, which always holds, this says
$\QBCE(G,Q)=\BCE(G,S_\emptyset)$ for every $G$, which by Theorem~\ref{thm:trivial} is
local flatness. The three conditions therefore agree.
\end{proof}

Theorems~\ref{thm:classicalcomplete} and~\ref{thm:bottomcomplete} therefore establish
equality of $\qis$ and $\inc$ on two natural classes of pairs: those that are
simultaneously diagonalizable in a product basis, and those in which one structure is
least informative. 



\section{Comparing information: a worked example}\label{sec:example}

Having established the foundations of the solution concept and the comparison theorem, we close by computing. The example exhibits a totally ordered family of quantum
information structures, verifies the comparative static of
Theorem~\ref{thm:soundness} on it, and shows what the ordering does to welfare. Every
quantity reported is the value of a semidefinite program in the sense of
Corollary~\ref{cor:spectra}.

\subsection{The basic game}

We use the binary investment game of~\cite[Appendix]{BM2016}, with their parameters, so
that the quantum computation can be read against the classical one. There are two
players, each choosing $a_i\in\{I,N\}$ (invest or not), and two states
$\Omega=\{\omega_B,\omega_G\}$ with $\psi(\omega_G)=\tfrac13$. Payoffs to player $1$
are
\[
\begin{array}{c|cc}
\omega_B & I & N\\\hline
I & z-1+y_B & -1\\
N & z & 0
\end{array}
\qquad\qquad
\begin{array}{c|cc}
\omega_G & I & N\\\hline
I & z+1+y_G & 1\\
N & z & 0
\end{array}
\]
with the first entry indexed by player $1$'s action, and symmetrically for player $2$.
We take $z=2$, $y_G=0$, $y_B=-\tfrac16$. The externality $z$ is large enough that joint
investment maximizes welfare in both states: the first-best is $a=(I,I)$ always, with
average payoff per player
$\tfrac23\cdot\tfrac56+\tfrac13\cdot3=\tfrac{14}9\approx1.5556$.

\subsection{A chain of quantum information structures}

Each player holds one qubit. Let
\[
\eta_{\omega_B}:=\proj0,\qquad \eta_{\omega_G}:=\proj+,
\]
two \emph{non-commuting} states, and let $\Lambda_\gamma$ be the qubit depolarizing channel
$\Lambda_\gamma(\rho)=\gamma\rho+(1-\gamma)\tfrac12\1$ for $\gamma\in[0,1]$. Define the family
\[
Q_\gamma:\qquad \rho^{(\gamma)}_\omega:=\Lambda_\gamma(\eta_\omega)\otimes\Lambda_\gamma(\eta_\omega),
\qquad \omega\in\{\omega_B,\omega_G\},
\]
so that the players hold conditionally independent noisy quantum samples of the state
of nature. At $\gamma=1$ the structure is the pair of pure non-commuting states; at $\gamma=0$
both $\rho^{(0)}_\omega$ equal $\tfrac14\1$ and the structure carries no information.

\begin{proposition}\label{prop:chainexample}
$Q_\gamma\garb Q_{\gamma'}$, and hence $Q_\gamma\qis Q_{\gamma'}$ and $Q_\gamma\inc Q_{\gamma'}$, whenever $1\ge \gamma\ge \gamma'\ge0$.
\end{proposition}

\begin{proof}
A direct computation gives $\Lambda_{s}\circ\Lambda_{s'}=\Lambda_{ss'}$:
$\Lambda_{s}(\Lambda_{s'}(\rho))=s\big(s'\rho+(1-s')\tfrac12\1\big)+(1-s)\tfrac12\1
=ss'\rho+(1-ss')\tfrac12\1$. For $\gamma'\le\gamma$ put $s:=\gamma'/\gamma\in[0,1]$ if
$\gamma>0$ (and $s$ arbitrary if $\gamma=0$, in which case $\gamma'=0$ and the claim is
trivial); then $\Lambda_{s}\circ\Lambda_\gamma=\Lambda_{\gamma'}$, so applying the
channel $\Lambda_{s}$ to each qubit carries $\rho^{(\gamma)}_\omega$ to
$\rho^{(\gamma')}_\omega$ for every $\omega$. This is a local
garbling in the sense of Definition~\ref{def:garb}(a); apply
Proposition~\ref{prop:instances}(a) and Theorem~\ref{thm:soundness}.
\end{proof}

By Theorem~\ref{thm:soundness} the sets $\QBCE(G,Q_\gamma)$ are therefore \emph{nested
increasing as $\gamma$ decreases}: more information means fewer equilibrium outcomes.

\subsection{The computation}

By Corollary~\ref{cor:spectra}, maximizing average welfare
\begin{equation}
  \sum_{\omega,a}\tfrac12\big(u_1(a,\omega)+u_2(a,\omega)\big)\nu(\omega,a)   
\end{equation}
over $\QBCE(G,Q_\gamma)$ is a semidefinite program in the variables $\{\rho_\omega^a\}$: eight $4\times4$ Hermitian blocks, two trace-normalization equalities, and eight
$2\times2$ Loewner constraints from Theorem~\ref{thm:loewner}, four per player. Solving
it gives the curve in Figure~\ref{fig:welfare}, computed on a grid of $41$ values of
$\gamma$.

\begin{figure}[h]
\centering
\includegraphics[width=\textwidth]{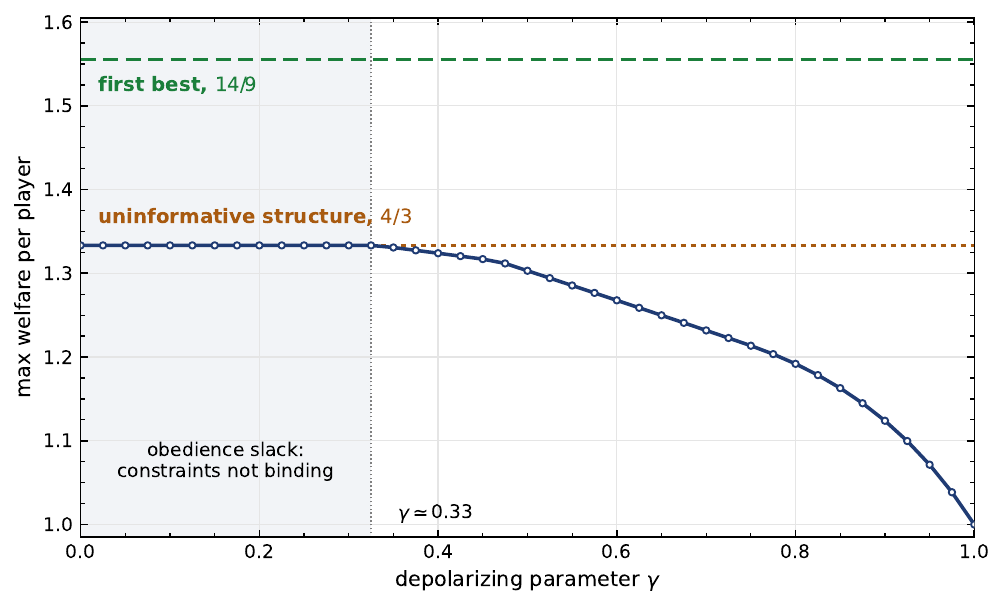}
\caption{Maximum average welfare per player over $\QBCE(G,Q_\gamma)$ along the depolarizing
chain of Proposition~\ref{prop:chainexample}, each marker being the optimal value of a
semidefinite program. Informativeness increases with $\gamma$, so the structures are ordered
from least informative on the left to most informative on the right. The value is
non-increasing, as Theorem~\ref{thm:soundness} requires: more information admits fewer
equilibrium outcomes and hence a lower best case. It equals $4/3$, the maximum over
$\BCE(G,S_\emptyset)$, for $\gamma\lesssim0.33$, where the obedience constraints are slack,
and falls to $1$ at $\gamma=1$. The first-best value $14/9$ is never attained, because joint
investment in the bad state is not obedient}\label{fig:welfare}
\end{figure}

Three features are worth recording.

\emph{The comparative static.} The maximum is non-increasing in $\gamma$, as
Proposition~\ref{prop:chainexample} and Theorem~\ref{thm:soundness} require: a larger
$\gamma$ means a more informative structure, hence a smaller QBCE set, hence a smaller
maximum. This is the quantum form of the comparative static that Bergemann and Morris obtain for
the same basic game. Their Figure~4 displays it as a nesting of BCE outcome sets in the
plane of state-conditional investment probabilities, drawn for three accuracies of the
classical signal, and they conclude from that nesting that the best achievable BCE
welfare is weakly lower with more information, and strictly lower in their
example~\cite[Appendix]{BM2016}. Figure~\ref{fig:welfare} plots the welfare conclusion
directly rather than the nesting from which it follows. The direction is the one Theorem~\ref{thm:soundness} forces: more information adds
obedience constraints and adds no feasibility, since realizability is free by
Lemma~\ref{lem:realize}, so the equilibrium set can only shrink and its best element can
only fall.

\emph{Consistency at the endpoints.} At $\gamma=0$ the value $\tfrac43$ coincides with the
maximum of welfare over $\BCE(G,S_\emptyset)$, computed separately as a linear program;
this is the value predicted by Theorem~\ref{thm:ppdegen}, since an uninformative
structure imposes only the null-structure constraints. At $\gamma=1$ the value falls to
$1$. The first-best $\tfrac{14}9$ is never attained, because joint investment in the
bad state is not obedient. This is the wedge that makes the multi-player comparative
static nontrivial, and its absence for a single agent is Proposition~\ref{prop:notvalue}.

\emph{A threshold.} The maximum is constant at $\tfrac43$ for $\gamma\lesssim0.33$ and
strictly decreasing thereafter. Below the threshold the quantum signals are too noisy
for the additional obedience constraints to bind at the welfare-maximizing outcome,
even though $Q_\gamma$ is strictly more informative than $Q_0$ in the order $\qis$. Strict
informativeness therefore need not translate into a strict welfare reduction: the
ordering is about set inclusion, and inclusions can be strict in directions that a
particular scalar objective does not detect.

\section{Concluding remarks}\label{sec:conclusion}
The transition from classical to quantum information structures in games requires a fundamental shift in how we conceptualize correlated advice. By abandoning the belief-invariant assumptions prevalent in prior quantum game theory and allowing the mediator to observe the state of nature, we construct a quantum analogue to the Bergemann--Morris framework that is both mathematically rigorous and conceptually conservative. The engine of this framework is the realization that quantum obedience translates exactly into a Loewner domination between operators on a single player's subsystem. From this single operator inequality, the core geometry of the theory follows: the quantum Bayes correlated equilibrium set is a convex, compact spectrahedron, and optimal information design reduces naturally to semidefinite programming.

This formulation successfully recovers the classical theory without distortion. An omniscient, coherent mediator generates no spurious outcomes when restricted to classical structures, proving that the quantum extension is exact. Furthermore, our analysis reveals a stark dichotomy in how quantum resources can be accessed in games. The "advice model"—where players extract classical signals from quantum states before equilibrium play—is entirely degenerate, carrying no incentive content beyond the null structure. The comparative power of a quantum information structure instead lies strictly in the local quantum deviations it permits. Under direct access, quantum individual sufficiency acts as a robust preorder encompassing local garblings, shared randomness, and shared entanglement, confirming the central comparative static: more quantum information monotonically shrinks the set of obedient outcomes.

We establish the exact equivalence of quantum individual sufficiency and the incentive ordering for two fundamental regimes. We prove that these relations coincide precisely for information structures that are simultaneously diagonalizable in a product basis, cleanly recovering the classical Bergemann--Morris equivalence. We further prove that this exact equivalence holds universally for any comparison involving the least informative structure. For all other quantum structures, we establish the rigorous soundness inclusion: quantum individual sufficiency guarantees the incentive ordering  across all basic games and any number of players. By resolving the exact bounds of the classical embedding and proving the general soundness of the quantum expansion, this framework places the comparison of quantum information structures on firm operational foundations.

\bibliographystyle{plain}
\bibliography{qbce-refs}

\end{document}